\RequirePackage{fix-cm}

\RequirePackage{amsmath,amssymb}

\documentclass[twocolumn,epjc3]{svjour3}% twocolumn

\RequirePackage[T1]{fontenc}

\smartqed  % flush right qed marks, e.g. at end of proof

\RequirePackage{graphicx}
\RequirePackage{mathptmx} % use Times fonts if available on your TeX system
\RequirePackage{flushend}

\RequirePackage{amsmath,amssymb}
\RequirePackage{color}
\RequirePackage{cite}
\RequirePackage[colorlinks,citecolor=blue,urlcolor=blue,linkcolor=blue]{hyperref}
\RequirePackage{cuted}
\newcommand\pd[2]{\frac{\partial #1}{\partial #2}}

\newcommand{\abs}[1]{\lvert#1\rvert}
\newcommand{\norm}[1]{\left\| #1 \right\|}

\journalname{Eur. Phys. J. C}
\begin{document}

\title{Partial spectral flow and the Aharonov--Bohm effect\\ in graphene 
}

\titlerunning{Partial spectral flow}        % if too long for running head

\author{Mikhail I.~Katsnelson\thanksref{e1,addr1}
        \and
        Vladimir Nazaikinskii\thanksref{e2,addr2,addr3} %etc.
}

\thankstext{e1}{e-mail: M.Katsnelson@science.ru.nl}
\thankstext{e2}{e-mail: nazaikinskii@googlemail.com}

\authorrunning{Mikhail I.~Katsnelson and Vladimir Nazaikinskii} % if too long for running head

\institute{Institute for Molecules and Materials, Radboud University, Heyendaalseweg 135, 6525AJ, Nijmegen, The Netherlands \label{addr1}
           \and
Ishlinsky Institute for Problems in Mechanics RAS, 101-1 Vernadsky Ave., Moscow 119526, Russia \label{addr2}
\and
Moscow Institute of Physics and Technology, Institutsky lane 9, Dolgoprudny, Moscow region 141700, Russia
\label{addr3}}

\date{Received: date / Accepted: date}

\maketitle

\begin{abstract}
We study the Aharonov--Bohm effect in an open-ended tube made of a graphene sheet whose dimensions are much larger than the interatomic distance in graphene. An external magnetic field vanishes on and in the vicinity of the graphene sheet and its flux through the tube is adiabatically switched on. It is shown that, in the process, the energy levels of the tight-binding Hamiltonian of $\pi$-electrons unavoidably cross the Fermi level, which results in the creation of electron--hole pairs. The number of pairs is proven to be equal to the number of magnetic flux quanta of the external field. The proof is based on the new notion of partial spectral flow, which generalizes the ordinary spectral flow already having well-known applications (such as the Kopnin forces in superconductors and superfluids) in condensed matter physics.

\keywords{Spectral flow \and Lattice fermion models \and Graphene \and Aharonov--Bohm effect \and Dirac equation \and Pair creation}
\end{abstract}
\maketitle
%\tableofcontents
% ----------------------------------------------------------------
\section{Introduction}

One of the main trends in contemporary theoretical physics and, in particular, theory of condensed matter is the increasing role of geometric and especially topological language \cite{Schap89,Thou,Naka,Volo,Kats12,Merm,Qi10,Hald,Kost}. Subtle and nontrivial topological effects in superfluid helium-3~\cite{Volo}, topologically protected zero-energy states in graphene in magnetic field~\cite{Kats12}, and the quickly growing field of topological insulators~\cite{Qi10} are just a few examples.

In most of cases, the use of topological concepts in condensed matter physics is closely related to the continuum-medium description. For example, the topology of electronic states in graphene, topological insulators, Weyl semimetals, and other ``topological quantum matter'' \cite{Hald} is studied for effective Hamiltonians describing the electronic band structure in the close vicinity of some special points in the Brillouin zone. In this approximation, the Hamiltonians are partial differential operators, and one can use the well-develop\-ed machinery, such as the concepts of index of Dirac operators \cite{ASind} or spectral flow \cite{APSSF}. Note that the appearance of nonzero spectral flow related to ``Dirac-like'' dynamics of fermions in the presence of vortices in rotating superfluid He-3 or in type II semiconductors leads to very interesting observable quantities such as additional forces acting on moving vortices \cite{BMCHHVV,Kop02,Volo,KoKr,Vol86,StGa87,KVP95,Vol13}. However, there also exist natural models in which  the Hamiltonians in periodic crystal lattices are matrices, and accordingly the Schr\"{o}dinger equation for electrons is a finite-difference equation rather than a differential one. Transfer of topological concepts to this case is in general a nontrivial mathematical problem. To our knowledge, it is a rather poorly studied field, at least, in the context of applications to condensed matter physics. Keeping in mind a broad use of lattice models in quantum field theory~\cite{Creu}, it may have even more general interest. Here we will give a solution of one particular problem of this kind, namely, a modification of the concept of spectral flow which is required when passing from the continuum-medium to lattice description of electronic structure of gra\-phene~\cite{Kats12}.

\begin{figure}[ht]
\centering
\includegraphics{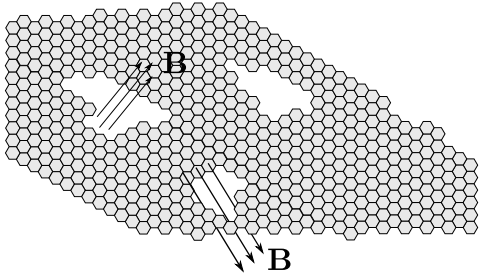}
\caption{Graphene flake}\label{fig01}
\end{figure}

Consider a flake with several holes containing magnetic fluxes (see Fig.~\ref{fig01}). Even when the magnetic field is nonzero only within the holes, it will affect the wave function and the energy spectrum of the electrons in the flake owing to the Aharonov--Bohm effect~\cite{AB59,OP85}. The spectrum should be a periodic function of the fluxes; namely, when all fluxes are changed by some integers (in the units of flux quantum), the spectrum should coincide with the initial one. If the Hamiltonian with purely discrete spectrum is bounded or at least semibounded above or below, it means automatically that the total number of, say, negative eigenvalues is a periodic function of the fluxes, and the spectral flow is zero.\footnote{Note, however, that the spectral flow occurring in the construction of the Kopnin spectral flow force~\cite{Kop02,Vol13} may well be nonzero even for a finite-dimensional Hamiltonian, because the periodicity condition is not satisfied there.} However, for the Dirac operator, which is unbounded on both sides, it can be also the shift of the spectrum, e.g., $E_n \longrightarrow E_{n+1}$. In this situation the spectral flow is nonzero. It was proven~\cite{Prokh,KatNa1} that such a situation arises in graphene for a certain kind of boundary conditions if the electrons in graphene are described by the Dirac approximation. This has important physical consequences~\cite{KatNa1}. In particular, a nonzero spectral flow means that for any position of the Fermi energy when changing the magnetic fluxes it will be unavoidably the situation when one of the energy levels coincides with the Fermi energy, which means all kind of specific many-body effects, potential instabilities, etc.~\cite{Kats12}.

However, literally speaking, this cannot be the case of real graphene, because the Dirac model is valid only within a close vicinity of the conical $K$ and $K'$ points. At larger energy scale, one needs to use a tight-binding model with a finite bandwidth \cite{Kats12}. Obviously, the usually defined spectral flow can be only zero in such a situation.

In this paper, we introduce a concept of \textit{partial} spectral flow for the tight-binding model of graphene. We will show that despite the vanishing of the total spectral flow the physical conclusion~\cite{KatNa1} on the unavoidable crossing of energy levels with the Fermi energy at adiabatically growing magnetic flux remains correct.

To make our consideration mathematically rigorous and to avoid unnecessary, purely technical
complications we will consider the situation simpler than in Fig.~\ref{fig01}, namely, a gra\-phene
tube (which can be considered as a carbon nanotube of a very large radius). We conjecture that the same situation takes place also for the case of graphene flake with several holes considered in~\cite{KatNa1}.

\section{Reminder: Hamiltonians of $\pi$-electrons in an infinite graphene sheet}\label{s2}

We use the common model described in~\cite[Chap.~1]{Kats12}. Recall that graphene has hexagonal (``honeycomb'') lattice with nearest-neighbor interatomic distance $a\approx1.42$\,\AA. The lattice naturally splits into two sublattices~$A$ and~$B$, where each atom in sublattice~$A$ is surrounded by three atoms of sublattice~$B$, and vice versa.
\begin{figure}[ht]
\centering
\includegraphics{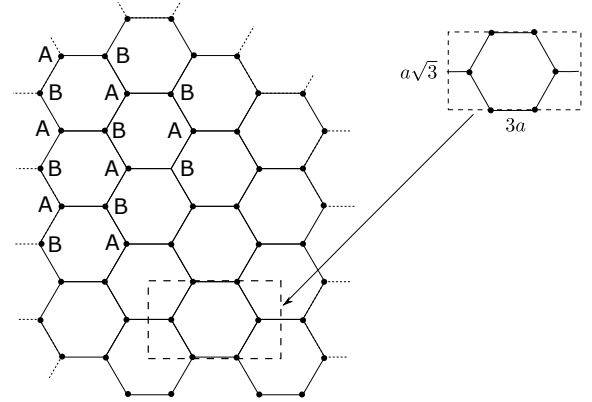}
\caption{Graphene honeycomb lattice and one of the hexagons}
\label{fig02}
\end{figure}
Geometrically, it will be convenient to us to think of the sheet plane as tiled by $3a\times\sqrt3a$ rectangles each containing a single hexagon of the lattice (see Fig.~\ref{fig02}). Each of sublattices~$A$ and~$B$ is a Bravais lattice with primitive vectors
\begin{equation*}
  a_1 =\biggl(\frac{3a}2,\frac{a\sqrt3}{2}\biggr),\quad
  a_2 =\biggl(\frac{3a}2,-\frac{a\sqrt3}{2}\biggr),
\end{equation*}
and the reciprocal lattice is generated by the vectors (see Fig.~\ref{fig03})
\begin{equation*}
  b_1 =\biggl(\frac{2\pi}{3a},\frac{2\pi}{a\sqrt3}\biggr),\quad
  b_2 =\biggl(\frac{2\pi}{3a},-\frac{2\pi}{a\sqrt3}\biggr).
\end{equation*}
\begin{figure}[t]
\centering
\includegraphics{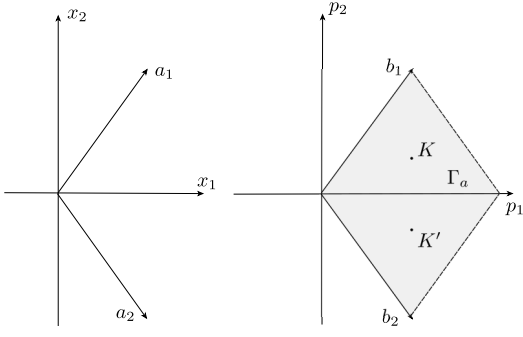}
\caption{Lattice vectors, reciprocal lattice vectors, the fundamental domain $\Gamma_a$, and the Dirac points $K$ and $K'$}
\label{fig03}
\end{figure}

In the tight-binding approximation, the electron $\psi$-func\-tion is defined on the lattice, and the Hamiltonian has the form
\begin{equation}\label{eq1-01}
  [\widehat H\psi](x)=\gamma_0\sum_y\psi(y),
\end{equation}
where the sum is over the three neighbors~$y$ of the lattice point~$x$ and $\gamma_0$ is a constant known as the \emph{hopping parameter}. Note that the sign of gamma does not affect any properties of the Hamiltonian and can be changed just by re-definition of the basis vectors~\cite[Chap.~1]{Kats12}. To be specific, we will assume here $\gamma_0 >0$.

The Hamiltonian can be conveniently expressed in terms of the operators $\widehat p=(\widehat p_1,\widehat p_2)$, $\widehat p_j=-i\pd{}{x_j}$, if the $\psi$-function is represented as a 2-vector $\psi=\bigl(\begin{smallmatrix}
                 \psi_B \\
                 \psi_A
               \end{smallmatrix}\bigr)$,
where $\psi_B$ and $\psi_A$ are the restrictions of~$\psi$ to sublattices~$B$ and~$A$, respectively. Then
\begin{equation}\label{tbh}
\begin{aligned}
\widehat H=H(\widehat p),\qquad H(p)&=
\gamma_0
\begin{pmatrix}
  0 & T(p) \\
  T^*(p) & 0 \\
\end{pmatrix},
\\
T(p)&=\sum_{j=1}^3e^{i\langle\delta_j,  p\rangle},
\end{aligned}
\end{equation}
where
\begin{equation*}
    \delta_1=\biggl(\frac a2,\frac{a\sqrt3}{2}\biggr),\quad
    \delta_2=\biggl(\frac a2,-\frac{a\sqrt3}{2}\biggr),\quad
  \delta_3=(-a,0)
\end{equation*}
are the vectors joining a point of sublattice~$B$ with its nearest $A$ neighbors
and $\langle u,v\rangle=u_1v_1+u_2v_2$. Thus, $e^{i\langle\delta_j,\widehat p\rangle}$ is a~shift operator,
\begin{equation}\label{shift}
\bigl[e^{i\langle\delta_j,\widehat p\rangle}\varphi\bigr](x)=\varphi(x+\delta_j).
\end{equation}
The function $T(p)$ vanishes at the \emph{Dirac points}
\begin{equation*}
  K=\biggl(\frac{2\pi}{3a},\frac{2\pi}{3a\sqrt3}\biggr),
  \qquad
 K'=\biggl(\frac{2\pi}{3a},-\frac{2\pi}{3a\sqrt3}\biggr)
\end{equation*}
of the reciprocal lattice (see Fig.~\ref{fig03}), and for $\psi$-functions localized in the momentum space near these points the Dirac Hamiltonians are used, which are obtained as approximations to the tight-binding Hamiltonian as follows. Make the change of variables
\begin{equation}\label{tbh2D}
   \begin{pmatrix}
     \psi_B(x) \\
     \psi_A(x) 
   \end{pmatrix}=
   W\begin{pmatrix}
     u_B(x) \\
     u_A(x)
   \end{pmatrix},\qquad
   W=e^{i\langle\widetilde K,x\rangle}
   \begin{pmatrix}
     1 & 0 \\
     0 & e^{-\tfrac{5\pi i}{6}} 
   \end{pmatrix},
\end{equation}
where $\widetilde K=K$ or~$K'$. Then the Hamiltonian acting on the vector functions $u=\bigl(\begin{smallmatrix} u_B\\ u_A \end{smallmatrix}\bigr)$ is
\begin{equation*}
  W^{-1}\widehat HW=\gamma_0\begin{pmatrix}
  0 & e^{-\tfrac{5\pi i}{6}}T(\widetilde K+\widehat p) \\
  e^{\tfrac{5\pi i}{6}}T^*(\widetilde K+\widehat p) & 0
\end{pmatrix}.
\end{equation*}
Assuming that $u_B(x)$ and $u_A(x)$ are smooth functions on~$\mathbb{R}^2$ varying slowly compared with the exponential $e^{i\langle\widetilde K,x\rangle}$, the symbol $T(\widetilde K+p)$ can be replaced in the first approximation by the linear part of its Taylor expansion at the point $p=0$, and
we obtain the Dirac Hamiltonian $\widehat D=D^{+}(\widehat p)$ if $\widetilde K=K$ or $\widehat D'=D^{-}(\widehat p)$ if $\widetilde K=K'$, where
\begin{align}\nonumber
  D^{\pm}(p)&=\frac{3a\gamma_0}2\begin{pmatrix}
      0 & p_1\pm ip_2 \\
      p_1\mp ip_2 & 0
    \end{pmatrix},
\\ \label{Dir}
  D^{\pm}(\widehat p)&=\frac{3a\gamma_0}2\begin{pmatrix}
      0 & \displaystyle -i\pd{}{x_1}\pm \pd{}{x_2} \\
      \displaystyle -i\pd{}{x_1}\mp\pd{}{x_2} & 0
    \end{pmatrix}.
\end{align}

\section{Main results}\label{ss22}

Consider a graphene tube in the shape of a right circular open-ended cylinder whose length and radius are both much greater than the distance between neighboring carbon atoms. We will study how the $\pi$-electron energy levels in graphene are affected if one adiabatically switches on a magnetic field $\mathbf{B}$ whose line pass through the tube and which vanishes on the tube surface (see Fig.~\ref{fig04}).
\begin{figure}[ht]
\centering
\includegraphics{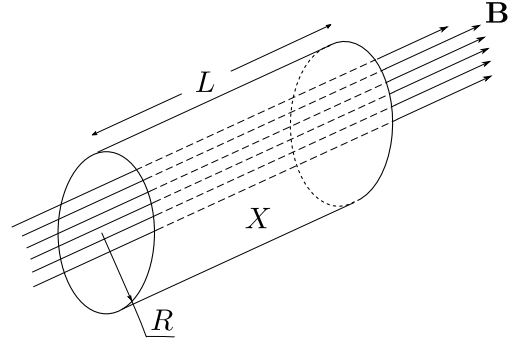}
\caption{Graphene tube~$X$}
\label{fig04}
\end{figure}

\subsection{Hamiltonians and boundary conditions}

We denote the cylinder by~$X$. Let~$L$ and~$R$ be the cylinder length and radius, respectively. We assume that $L\gg a$ and~$R\gg a$, where $a$ is the nearest-neighbor interatomic distance.
The circumference of the tube is~$l=2\pi R$. We use the coordinates $(x_1,x_2)$ on~$X$, where $x_1\in[0,L]$ is the coordinate along the cylinder axis and $x_2\in[0,l]$ is the circumferential coordinate (so that the endpoints of $[0,l]$ are glued together) and sometimes identify $X$ with $[0,L]\times[0,l]$.

The unfolded graphene tube is shown in Fig.~\ref{fig05}. We assume that the graphene lattice, which we denote by~$X_a=X_A\cup X_B\subset X$, has zigzag boundaries at the tube ends.
\begin{figure}[ht]
\centering
\includegraphics{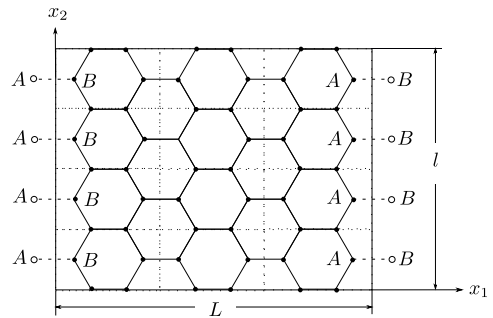}
\caption{Unfolded graphene tube ($M=3$ and $N=4$). The lattice points shown by open dots outside the rectangle $[0,L]\times[0,l]$ are fictitious; i.e., the corresponding carbon atoms are not actually present in the tube}
\label{fig05}
\end{figure}
Then we have $L=3aM$ and $l=\sqrt3aN$, where $M$ and $N$ are the numbers of elementary $3a\times\sqrt3a$ rectangles (cf. Fig.~\ref{fig02}, right)  along the $x_1$- and $x_2$-axis, respectively. It is easily seen that the lattice $X_a$ has $4MN$ vertices. Mathematically, it is convenient to assume that $L$ and $l$ are constant and $a$ is a small parameter. Thus, $a\to0$ and accordingly $M,N\to\infty$ so that the ratio $M/N$ remains constant, $M/N=L/(l\sqrt3)$. This is the point of view we take in what follows. 

For the graphene tube, definition~\eqref{eq1-01} (or, equivalently, \eqref{tbh}) of the tight-binding Hamiltonian fails to work at the boundary sites, where one of the neighboring lattice points is missing (see Fig.~\ref{fig05}). To make the definition work, we must somehow define the values of the $\psi$-function at the ``fictitious'' neighboring sites outside~$X$ based on its values at the sites belonging to~$X_a$. There are many ways to do this; here we use the simplest rule and define the value at an outer site to be equal to the value at the nearest inner site; i.e., we set
\begin{equation}\label{bc}
\begin{aligned}
  \psi_A\Bigl(-\frac a2,x_2\Bigr)&:=\psi_B\Bigl(\frac a2,x_2\Bigr),\\
  \psi_B\Bigl(L+\frac a2,x_2\Bigr)&:=\psi_A\Bigl(L-\frac a2,x_2\Bigr).
\end{aligned}
\end{equation}
A straightforward computation shows that the operator $\widehat H$ defined by~\eqref{eq1-01} with the boundary conditions~\eqref{bc} is self-ad\-joint in the Hilbert space ${\mathcal{H}}_a=\ell^2(X_a)$ with inner product
\begin{equation}\label{ipXA}
  (\psi,\widetilde\psi)=\frac{1}{4MN}\sum_{x\in X_a}\overline{\psi(x)}\widetilde\psi(x).
\end{equation}

Now if we substitute~\eqref{tbh2D} into~\eqref{bc} and let $a\to0$, then we arrive at the boundary conditions for the Dirac operators~\eqref{Dir}. They have the form
\begin{equation}\label{bmbc}
  -iu_B(0,x_2)=u_A(0,x_2),\quad
   -iu_B(L,x_2)=u_A(L,x_2)
\end{equation}
and are a special case of the \emph{Berry--Mondragon boundary conditions}~\cite{BeMo}
\begin{equation*}
 (n_{x_2}-in_{x_1})u_B=\varkappa u_A,
\end{equation*}
where $\mathbf{n}=(n_{x_1},n_{x_2})$ is the inward normal on the boundary and $\varkappa$ is a nonvanishing real-valued function on the boundary. Indeed, $\mathbf{n}=(1,0)$ at the left end of the tube ($x_1=0$), and $\mathbf{n}=(-1,0)$ at the right end ($x_1=L$). Thus, $\varkappa=1$ for the first condition in~\eqref{bmbc}, and $\varkappa=-1$ for the second condition.  The expressions~\eqref{Dir} with the boundary conditions~\eqref{bmbc} define self-adjoint operators~$\widehat D$ and~$\widehat D'$ on the Hilbert space ${\mathcal{H}}_0=L^2(X)\oplus L^2(X)$ with inner product
\begin{equation}\label{ipX}
  (u,v)=\frac{1}{2Ll}\iint_{[0,L]\times[0,l]}
  \bigr(\overline{u_A(x)}v_A(x)+\overline{u_B(x)}v_B(x)\bigl)\,dx.
\end{equation}

\subsection{Switching on the magnetic field}

Consider a magnetic field~$\mathbf{B}$ vanishing on and in the vicinity of the tube surface. (This is the setting in which one speaks of the Aharonov--Bohm effect: the field is zero in the domain where the particles (in our case, the $\pi$-electrons) are confined. However, note that all the subsequent constructions remain valid under the weaker condition that the normal component of~$\mathbf{B}$ vanishes everywhere on the tube surface.) Let us switch on the field adiabatically. This means that we have a continuous family $\mathbf{B}(t)$ of magnetic fields vanishing on~$X$ such that $\mathbf{B}(0)=0$ and $\mathbf{B}(1)=\mathbf{B}$, and $t$ is slow (``adiabatic'') time; that is, $t$ varies with the ordinary time so slowly that the system can be viewed as passing through a family of stationary states. Physically, this means that the dissipation of the energy levels due to the finite time of the process must be much less than the distance between neighboring energy levels, $\hbar/\tau\ll\Delta E$, where $\hbar$ is the Planck constant, $\tau$ is the actual (physical) time of the switching-on process, and $\Delta E$ is the interlevel distance (which in our problem is of the order of the hopping parameter $\gamma_0$ divided by the sample area, that is, of the order of $\gamma_0/(MN)$).  The simplest example is $\mathbf{B}(t)=t\mathbf{B}$. We can write $\mathbf{B}(t)=\nabla\times\mathbf{A}(t)$, where $\mathbf{A}(t)$ is the magnetic vector potential. It will be assumed without loss in generality that $\mathbf{A}(0)=0$ (which is consistent with the condition $\mathbf{B}(0)=0$).
Let $A_1(x,t)$ and $A_2(x,t)$, $x\in X$, be the axial and circumferential components, respectively, of the vector potential~$\mathbf{A}(t)$ restricted to the tube surface. We write ${\mathrm{A}}=(A_1,A_2)$. (If magnetic potentials are interpreted as differential $1$-forms, then $A_1(x,t)\,dx_1+A_2(x,t)\,dx_2$ is just the restriction of $\mathbf{A}(t)$ to~$X$.)

The condition that $\mathbf{B}(t)=0$ on~$X$ implies that
\begin{equation}\label{e-zerocurl}
  \pd{A_1}{x_2}-\pd{A_2}{x_1}=0,\qquad x\in X.
\end{equation}
In the presence of the magnetic field~$\mathbf{B}(t)$, the boundary conditions remain the same, and the momentum operator occurring in the Hamiltonians is modified as follows \cite[Ch.~2]{Kats12}:
\begin{equation}\label{longmomenta}
  \widehat p_j=-i\pd{}{x_j}\longmapsto
  \widehat p_j-A_j(x,t),\qquad j=1,2.
\end{equation}
(We work in a system of units where $e=1$ and $c=1$ and omit the factor $e/c$.) Thus, in the Dirac approximation we have the Hamiltonians
\begin{equation}\label{Dir-m}
 \widehat D_t=D^+(\widehat p-{\mathrm{A}}(x,t)), \qquad
 \widehat D_{\mathbf{A}}'=D^-(\widehat p-{\mathrm{A}}(x,t))
\end{equation}
corresponding to the $K$ and $K'$ valleys, respectively, with the boundary conditions~\eqref{bmbc}, and the  tight-binding Hamiltonian becomes
\begin{equation}\label{tbh-m}
   \widehat H_t=H(\widehat p-{\mathrm{A}}(x,t))
\end{equation}
with the boundary conditions~\eqref{bc}. The symbol $H(p)$ (see~\eqref{tbh}) involves exponential functions of~$p$, and so it might be helpful if we explain how the right-hand side of~\eqref{tbh-m} is defined. It suffices to define the exponential $e^{i\langle \delta_j,\widehat p-{\mathrm{A}}(x,t)\rangle}$. This exponential is none other than the value at $\tau=1$ of the solution of the Cauchy problem for the first-order differential equation
\begin{equation*}
  -i\pd u\tau=\langle \delta_j,\widehat p-{\mathrm{A}}(x,t)\rangle u,\qquad
  u|_{t=0}=1.
\end{equation*}
By solving this problem, we find that
\begin{multline}\label{shiftm}
  e^{i\langle \delta_j,\widehat p-{\mathrm{A}}(x,t)\rangle}
  \\=
  \exp\Bigl\{-i\int_{0}^{1}\langle\delta_j,
  {\mathrm{A}}(x+\tau\delta_j,t)\rangle \,d\tau\Bigr\}e^{i\langle \delta_j,\widehat p\rangle}.
\end{multline}

Now assume that the magnetic flux~$\Phi$ of the field~$\mathbf{B}$ through the tube is an integer multiple of $2\pi$:
\begin{equation}\label{Phi}
  \Phi=\int_{0}^{l} A_2(x_1,x_2,1)\,dx_2=2\pi q,\qquad q\in\mathbb{Z}.
\end{equation}
(The integral in~\eqref{Phi} is independent of~$x_1$ by condition~\eqref{e-zerocurl}.) The number~$q$ is referred to as the ``number of magnetic flux quanta'' through the tube. In view of~\eqref{e-zerocurl}, there exists a function $S(x)$ on the rectangle $[0,L]\times[0,l]$ such that $\nabla S(x)=\mathrm{A}(x,1)$, and it follows from~\eqref{Phi} that
\begin{equation*}
  S(x_1,l)-S(x_1,0)=2\pi q.
\end{equation*}
Consequently, $e^{iS(x_1,0)}=e^{iS(x_1,l)}$, the formula
\begin{equation}\label{U}
  U(x)=e^{iS(x)}
\end{equation}
gives a well-defined smooth function on the cylinder~$X$, and one has
\begin{equation*}
  \nabla U(x)=A(x,1)U(x).
\end{equation*}
It follows that $\widehat p-{\mathrm{A}}(x,1)=U \widehat p U^{-1}$, and we see that the gauge transformation by~$U$ establishes a unitary equivalence between the Hamiltonians at $t=0$ and $t=1$:
\begin{equation}\label{UU-1}
\begin{aligned}
  \widehat H&\equiv \widehat H_0=U^{-1} \widehat H_1 U,\\
   \widehat D\equiv \widehat D_0=U \widehat D_1 U^{-1},&\qquad
   \widehat D'\equiv \widehat D_0'=U \widehat D_1' U^{-1}.
\end{aligned}
\end{equation}
Thus, the spectrum of each of these Hamiltonians without the magnetic field is the same as that of the same Hamiltonian with the magnetic field fully switched on. But what happens with the spectrum in between, that is, as $t$ varies from~$0$ to~$1$? Do the eigenvalues cross the zero level? How many of them do so, and in what direction?

\subsection{Aharonov--Bohm effect for the Dirac Hamiltonians}\label{ss32}

The answer for the case of Dirac Hamiltonians was given in~\cite{Prokh,KatNa1}. An adequate tool for describing the motion of eigenvalues is given by the notion of \emph{spectral flow} introduced by Atiyah, Patodi, and Singer~\cite{APSSF}, which can be informally described as follows. Consider a family $\{B_t\}_{t\in[0,1]}$ of self-adjoint operators that in some sense continuously depend on~$t$ and whose spectrum in a neighborhood of zero is purely discrete. Then the spectral flow $\operatorname{sf}\{B_t\}$ is the net number of eigenvalues crossing zero in the positive direction as $t$ varies from~$0$ to~$1$ (see Fig.~\ref{fig06}).
\begin{figure}[ht]
\centering
\includegraphics{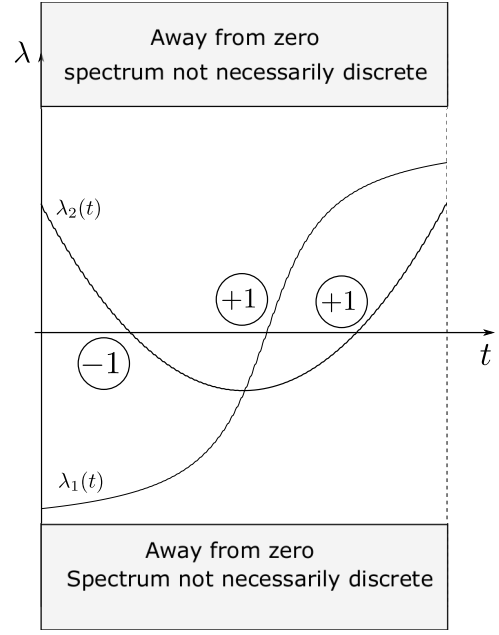}
\caption{Example of computation of the spectral flow. The eigenvalue~$\lambda_1(t)$ has one crossing counting as~$+1$; th eigenvalue $\lambda_2(t)$ has two crossings, one counting as~$-1$ and the other as~$+1$; as a result, the spectral flow is~$(+1)+(-1)+(+1)=1$}
\label{fig06}
\end{figure}
The rigorous definition can be found in~\cite{APSSF} and, in a different form, in~\cite{BLP1} (see also~\cite{NSScS99} and Remark~\ref{rk-sf} in the next subsection).
The spectral flow is homotopy invariant in the class of families such that $B_0$ and $B_1$ are isospectral (i.e., have the same spectrum) and hence can be computed by topological means. A formula for the spectral flow of Dirac Hamiltonians on an arbitrary graphene ``flake'' was conjectured in~\cite{Prokh} and then shown to be true in~\cite{KatNa1}, where a general theorem on the spectral flow of families of Dirac type operators with classical boundary conditions on a compact manifold with boundary was proved. In our situation, this formula is as follows.
\begin{proposition}[\textnormal{special case of~{\cite[Theorem~1]{KatNa1}}}]\label{Th1}
Let condition~\eqref{Phi} be satisfied.
Then the spectral flow of the families~\eqref{Dir-m} is given by the formula
\begin{equation}\label{spfl-0}
  \operatorname{sf}\{\widehat D_t'\}=-\operatorname{sf}\{\widehat D_t\}=q.
\end{equation}
\end{proposition}
Thus, the spectral flow coincides (up to the sign) with the number of magnetic flux quanta. 

\subsection{Partial spectral flow}

If we try to apply the same tool---spectral flow---to the case of the tight-binding Hamiltonian, then we immediately see that such an approach fails. Indeed, the tight-binding Hamiltonian acts on the finite-dimensional space~${\mathcal{H}}_a$, and hence the spectral flow of the family~$\widehat H_t$ (as well as of any operator family $\{B_t\}$ on a finite-dimensional space with isospectral~$B_0$ and~$B_1$) is necessarily zero.

That is why we introduce a finer notion of \emph{partial spectral flow along a subspace}, which takes into account not only the eigenvalues themselves but also how close the corresponding eigenvectors are to a given subspace.

Let ${\mathcal{H}}$ be a Hilbert space, and let ${\mathcal{L}}\subset{\mathcal{H}}$ be a (closed) subspace. The orthogonal projection onto~${\mathcal{L}}$ in~${\mathcal{H}}$ will be denoted by~$P_{\mathcal{L}}$.

Consider a family $\{B_t\}$, $t\in[0,1]$, of self-adjoint operators on~${\mathcal{H}}$. By $E(B_t,J)$, where $J\subset\mathbb{R}$ is an arbitrary interval, we denote the orthogonal projection in~${\mathcal{H}}$ onto the closed linear span of eigenvectors of~$B_t$ corresponding to the eigenvalues lying in~$J$.

\begin{definition}\label{d-Ltame}
The family $\{B_t\}$ is said to be ${\mathcal{L}}$-\emph{tame} if the following conditions are satisfied:
\begin{enumerate}
\item[(i)] The resolvent $(i-B_t)^{-1}$ continuously depends on~$t\in[0,1]$ in the operator norm.
\end{enumerate}
Next, there exists a $\delta>0$ such that
\begin{enumerate}
  \item[(ii)] For each $t\in[0,1]$, the spectrum of~$B_t$ on the interval $(-\delta,\delta)$ is purely discrete.
  \item[(iii)] For any $t\in[0,1]$ and any interval $J\subset(-\delta,\delta)$, one has
\begin{equation}\label{comm-norm}
  \norm{[P_{\mathcal{L}},E(B_t,J)]}<\frac14.
\end{equation}
\end{enumerate}
Here $[P_{\mathcal{L}},E(B_t,J)]=P_{\mathcal{L}} E(B_t,J)-E(B_t,J)P_{\mathcal{L}}$ is the commutator of~$P_{\mathcal{L}}$ and~$E(B_t,J)$.
\end{definition}

Let $\{B_t\}$, $t\in[0,1]$, be an ${\mathcal{L}}$-tame family.
By~(i) and~(ii), for some~$n$ there exists a~partition $0=t_0<t_1<t_2<\dotsm<t_{n+1}=1$ of the interval $[0,1]$ and numbers $\gamma_1,\dotsc,\gamma_{n+1}\in(-\delta,\delta)$ such that $\gamma_j$ does not lie in the spectrum $\operatorname{Spec}(B_t)$ of the operator $B_t$ for $t\in[t_{j-1},t_j]$, $\gamma_1=\gamma_{n+1}\le 0$, and if $\gamma_1<0$, then the half-open interval $[\gamma_1,0)$ does not contain any points of spectrum of $B_0$ and $B_1$. Let ${\mathcal{V}}_j={\mathcal{V}}(B_{t_j},\gamma_j,\gamma_{j+1})$ be the linear span of eigenvectors of~$B_{t_j}$ corresponding to the eigenvalues lying between~$\gamma_j$ and~$\gamma_{j+1}$. On the subspace ${\mathcal{V}}_j$, consider the quadratic form
\begin{equation}\label{quad-form}
  A_j[u]=(u,(2P_{\mathcal{L}}-1)u),\quad u\in {\mathcal{V}}_j.
\end{equation}
Let $m_{j+}=\sigma_+(A_j)$ be the positive index of inertia of the form~\eqref{quad-form}, i.e., the dimension of the positive subspace ${\mathcal{V}}_{j+}\subset {\mathcal{V}}_j$ of this form.

\begin{definition}\label{def-spf}
The \textit{partial spectral flow} of the ${\mathcal{L}}$-tame family $\{B_t\}$, $t\in[0,1]$,
\emph{along}~${\mathcal{L}}$ is the number
\begin{equation}\label{e06}
    \operatorname{sf}_{\mathcal{L}}\{B_t\}
    =\sum_{j=1}^nm_{j+}\operatorname{sign}(\gamma_j-\gamma_{j+1}).
\end{equation}
\end{definition}

\begin{remark}\label{rk-sf}
The definition of ``traditional'' spectral flow in the form given in~\cite{BLP1,NSScS99} is the special case of Definition~\ref{def-spf} for ${\mathcal{L}}={\mathcal{H}}$. Here condition~(iii) in Definition~\ref{d-Ltame} is satisfied automatically, and the numbers~$m_{j+}$ become the dimensions~$m_j$ of the eigenspaces~${\mathcal{V}}_j$.

For the general case of ${\mathcal{L}}\subsetneq{\mathcal{H}}$, the subspace~${\mathcal{V}}_{j+}$ can be thought of as the part of~${\mathcal{V}}_j$ ``close'' to the subspace~${\mathcal{L}}$.
\end{remark}

Some properties of the partial spectral flow are stated in the following theorem.

\begin{theorem}\label{th1}
\textup{(a)} Let $\{B_t\}$, $t\in[0,1]$, be an ${\mathcal{L}}$-tame family of self-adjoint operators. The  partial spectral flows $\operatorname{sf}_{\mathcal{L}}\{B_t\}$ and $\operatorname{sf}_{{\mathcal{L}}^\perp}\{B_t\}$ are well defined, and
\begin{equation}\label{compl}
    \operatorname{sf}_{\mathcal{L}}\{B_t\}+\operatorname{sf}_{{\mathcal{L}}^\perp}\{B_t\}=\operatorname{sf}\{B_t\}.
\end{equation}
\textup{(b) (}homotopy invariance of the partial spectral flow\textup) Let $\{B(t,\tau)\}$ be a two-parameter family of self-adjoint operators satisfying conditions~\textup{(i)--(iii)} in Definition~\textup{\ref{d-Ltame}} in which $t\in[0,1]$ is everywhere replaced with $(t,\tau)\in[0,1]\times[0,1]$. If
\begin{equation}\label{weak-iso}
    \operatorname{sf}_{\mathcal{L}}\{B(0,t)\}=
    \operatorname{sf}_{\mathcal{L}}\{B(1,t)\},
\end{equation}
then
\begin{equation*}
    \operatorname{sf}_{\mathcal{L}}\{B(t,0)\}=
    \operatorname{sf}_{\mathcal{L}}\{B(t,1)\}.
\end{equation*}
\end{theorem}

The proof of this theorem, as well as some more details concerning the partial spectral flow, will be given in Sec.~\ref{s4}.

\subsection{Aharonov--Bohm effect for the tight-binding Hamiltonian}

Here we will show that although the spectral flow $\operatorname{sf}\{\widehat H_t\}$ is zero, there is nonetheless a nontrivial motion of eigenvalues as $t$ varies from~$0$ to~$1$. Namely, there exist subspaces ${\mathcal{L}},{\mathcal{L}}'\subset{\mathcal{H}}_a$ consisting of functions localized in the momentum space near the Dirac points~$K$ and~$K'$, respectively, and such that the partial spectral flows of the family $\{\widehat H_t\}$ along these subspaces coincide with the spectral flows~\eqref{spfl-0} of the respective families of Dirac operators.

Our first task will be to define these subspaces, and to this end we introduce a basis in~${\mathcal{H}}_a$. Consider the set~$G_0$ of pairs $(m,n)$ of integers such that

(a)~$-N\le n\le N-1$;

(b)~$M+1\le m\le 3M-1$ if $-N\le n\le-N/2$ or $N/2<n\le N-1$;

(c)~$M\le m\le 3M$ if $-N/2<n\le N/2$.

It is easily seen that $G_0$ contains exactly $4MN$ elements.
\begin{lemma}[\textnormal{see Sec.\!~\ref{ssA1} for the proof}]\label{l202}
The functions
\begin{equation}\label{basis1}
  \varphi_{mn}(x)=
  \begin{cases}
    e^{i\tfrac{\pi m}L x_1
    +i\tfrac{2\pi n}l x_2},& x\in X_B,\\
    e^{-i\tfrac{\pi m}L x_1
    +i\tfrac{2\pi n}l x_2},& x\in X_A,
  \end{cases}
  \quad (m,n)\in G_0,
\end{equation}
form an orthonormal basis in ${\mathcal{H}}_a$.
\end{lemma}

To simplify the exposition, we will assume that $N$ is a multiple of~$3$. Set
\begin{equation*}
  \overline{m}=2M,\quad \overline{n}=\frac N3.
\end{equation*}
Note that the $\varphi_{mn}$ can be rewritten in the form
\begin{equation}\label{basis-K}
\begin{aligned}
  \varphi_{mn}(x)&=e^{i\langle K,x\rangle}
  \begin{cases}
    e^{i\tfrac{\pi(m-\overline{m})}L x_1
    +i\tfrac{2\pi(n-\overline{n})}l x_2},& x\in X_B,\\
    e^{\tfrac{2\pi i}{3}}e^{i\tfrac{\pi(\overline{m}-m)}L x_1
    +i\tfrac{2\pi(n-\overline{n})}l x_2},& x\in X_A,
  \end{cases}
\\
&=e^{i\langle K',x\rangle}
  \begin{cases}
    e^{i\tfrac{\pi(m-\overline{m})}L x_1
    +i\tfrac{2\pi(n+\overline{n})}l x_2},& x\in X_B,\\
    e^{\tfrac{2\pi i}{3}}e^{i\tfrac{\pi(\overline{m}-m)}L x_1
    +i\tfrac{2\pi(n+\overline{n})}l x_2},& x\in X_A.
  \end{cases}
\end{aligned}
\end{equation}
Thus, the function $\varphi_{mn}$ with $m=\overline{m}$ and $n=\overline{n}$ (or $n=-\overline{n}$) is just the exponential $e^{i\langle K,x\rangle}$ (or $e^{i\langle K',x\rangle}$) with the additional phase factor $e^{\tfrac{2\pi i}{3}}$ on sublattice~$A$. Accordingly, the $\varphi_{mn}$ with $(m,n)$ close to $(\overline{m},\pm\overline{n})$ are localized in the momentum space near the Dirac points~$K$ and~$K'$.

Take some $d>0$ and define subspaces~${\mathcal{L}},{\mathcal{L}}'\subset {\mathcal{H}}_a$ as the linear spans
\begin{align}\label{cL}
  {\mathcal{L}} &=\operatorname{Lin}\{\varphi_{mn}\colon (m-\overline{m})^2+(n-\overline{n})^2\le d^2\},
\\ \label{cLp}
  {\mathcal{L}}' &=\operatorname{Lin}\{\varphi_{mn}\colon (m-\overline{m})^2+(n+\overline{n})^2\le d^2\}.
\end{align}
The domains corresponding to~${\mathcal{L}}$ and~${\mathcal{L}}'$ in the momentum space are shown in Fig.~\ref{fig07}.

Now we are in a position to state the main theorem of the present paper.

\begin{theorem}\label{th-main}
There exists a $d>0$ \textup(which may depend on the family~$\mathbf{B}(t)$\textup) such that, for all sufficiently small $a>0$, the family~$\widehat H_t$ is ${\mathcal{L}}$-, ${\mathcal{L}}'$-, and $({\mathcal{L}}\oplus{\mathcal{L}}')$-tame, and
\begin{align*}
  \operatorname{sf}_{{\mathcal{L}}}\{\widehat H_t\} &=\operatorname{sf} \{\widehat D_t\}, \qquad
  \operatorname{sf}_{{\mathcal{L}}'}\{\widehat H_t\} =\operatorname{sf} \{\widehat D'_t\},
\\
   &\operatorname{sf}_{({\mathcal{L}}
   \oplus{\mathcal{L}}')^\perp}\{\widehat H_t\}=0.
\end{align*}
\end{theorem}
\begin{figure}[ht]
\centering
\includegraphics{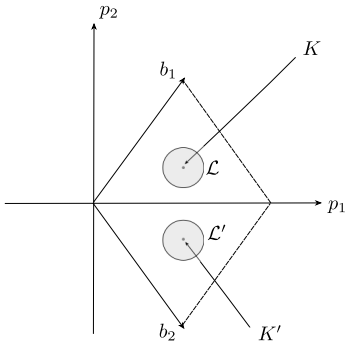}
\caption{Domains in the momentum space corresponding to the subspaces~${\mathcal{L}}$ and~${\mathcal{L}}'$ (shown by dashed disks)}
\label{fig07}
\end{figure}

Thus, informally speaking, all nontrivial spectral flow in concentrated near the Dirac points~$K$ and~$K'$ in the momentum space, and the partial spectral flows of the tight-binding Hamiltonian near these points are equal to the spectral flows provided by the respective Dirac approximations.

\subsection{Proof of Theorem~\ref{th-main}}\label{ss34}

We will only prove the assertion of the theorem for the subspace~${\mathcal{L}}$. The proof for the subspace~${\mathcal{L}}'$ is, mutatis mutandis, essentially the same. As to the claim for the subspace~$({\mathcal{L}}\oplus{\mathcal{L}}')^\perp$, it readily follows from Lemmas~\ref{EVP-est} and~\ref{SOS} below; we omit the details.

To make the proof more readable, we have transferred some technical computations to~\ref{AppA}.

\paragraph{\textbf{A.}}  First, note that the specific value of~$\gamma_0$
does not affect the assertion of the theorem in any way, because the spectral flow, as well as the partial spectral flow, does not change if the operator family is multiplied by a positive constant. Thus, we can take any $\gamma_0>0$ convenient to us instead of the actual, physically meaningful value, and from now on we set $\gamma_0=\frac2{3a}$ so as to ensure that the factor $\frac{3a\gamma_0}2$ occurring in formulas~\eqref{Dir} for the Dirac operators is equal to unity.

\paragraph{\textbf{B.}} Let $\Phi(t)=2\pi q(t)$ be the flux of the field~$\mathbf{B}(t)$ through the tube, $q(0)=0$, $q(1)=q\in\mathbb{Z}$. The potentials ${\mathrm{A}}(x,t)$ and ${\mathrm{A}}_0(t)=(0,\Phi(t) l^{-1})$ (the latter being independent of~$x$) generate the same flux, and hence there exists a smooth real-valued function $F(x,t)$ on $X\times[0,1]$ such that $\nabla_x F={\mathrm{A}}-{\mathrm{A}}_0$. The corresponding gauge transformation $\psi\mapsto U_t^{-1}\psi$, where $U_t$ is the operator of multiplication by $e^{iF(x,t)}$, reduces the family $\widehat H_t=H(\widehat p-{\mathrm{A}}(x,t))$ to the family $\widehat H_{0t}=H(\widehat p-{\mathrm{A}}_0(t))$ of operators with constant magnetic potential:
\begin{equation}\label{sim-tbh}
  \widehat H_t
  =U_t H(\widehat p-{\mathrm{A}}_0(t)) U_t^{-1}\equiv
  U_t \widehat H_{0t} U_t^{-1}.
\end{equation}

\paragraph{\textbf{C.}} It follows from~\eqref{sim-tbh} that any eigenvector of~$\widehat H_t$ has the form $U_t\psi$, where $\psi$~is an eigenvector of~$\widehat H_{0t}$ with the same eigenvalue. Let us study the eigenvalue problem for the operator~$\widehat H_{0t}$. The operator~$\widehat H_{0t}$ acts on the basis vectors $\varphi_{mn}$ by the formulas
\begin{align}\label{Hphi}
  \widehat H_{0t}\varphi_{mn}&=\mu(m,n,t)\varphi_{2\overline{m}-m,n},
\\\label{Hphi-ast}
  \widehat H_{0t}\varphi_{2\overline{m}-m,n}&=\mu(2\overline{m}-m,n,t)\varphi_{mn}
  =\mu^*(m,n,t)\varphi_{mn},
\end{align}
$(m,n)\in G_0$. (These formulas are proved in Sec.\!~\ref{ssA2}, where we give explicit expressions for~$\mu(m,n,t)$.) It follows from \eqref{Hphi} and~\eqref{Hphi-ast} that~${\mathcal{H}}_a$ splits into the orthogonal direct sum of two-dimensional invariant subspaces
\begin{equation*}
{\mathcal{V}}_{mn}=\operatorname{Lin}\{\varphi_{mn},\varphi_{2\overline{m}-m,n}\},\qquad
(m,n)\in G_0,\quad m>\overline{m},
\end{equation*}
and one-dimensional invariant subspaces
\begin{equation*}
{\mathcal{W}}_n=\operatorname{Lin}\{\varphi_{\overline{m} n}\},\qquad
-N\le n\le N-1.
\end{equation*}
On the subspace~${\mathcal{V}}_{mn}$, the operator~$\widehat H_{0t}$ is represented by the $2\times 2$ antidiagonal matrix with antidiagonal entries $\mu(m,n,t)$ and $\mu^*(m,n,t)$, and hence
the eigenvalues of~$\widehat H_{0t}$ on~${\mathcal{V}}_{mn}$ are $\pm\abs{\mu(m,n,t)}$. The eigenvalue of~$\widehat H_{0t}$ on~${\mathcal{W}}_n$ is $\mu(\overline{m},n,t)$.

\begin{lemma}[\textnormal{see Sec.\!~\ref{ssA3} for the proof}]\label{EVP-est}
There exists numbers $\delta,\overline{q},a_0>0$ such that if $a<a_0$ and $\psi$ is an eigenvector of~$\widehat H_{0t}$ with eigenvalue $\lambda$ satisfying $-\delta<\lambda<\delta$, then
\begin{equation*}
  \psi\in\bigoplus_{j=-\overline{q}}^{\overline{q}} \bigl({\mathcal{W}}_{j+\overline{n}}\oplus{\mathcal{W}}_{j-\overline{n}}\bigr).
\end{equation*}
\end{lemma}

\paragraph{\textbf{D.}} We need to prove that the family $\widehat H_{0t}$ of self-adjoint operators is ${\mathcal{L}}$-tame for sufficiently large~$d$. Conditions~(i) and~(ii) in Definition~\ref{d-Ltame} are trivially satisfied, because the family continuously depends on~$t$ and acts on the finite-dimensional space~${\mathcal{H}}_a$. To verify~(iii), take an arbitrary interval $J\subset(-\delta,\delta)$ and fix a $t\in[0,1]$. The orthogonal projection $E(\widehat H_t,J)$ onto the linear span of eigenvectors of~$\widehat H_t$ corresponding to the eigenvalues lying in~$J$ has the form
\begin{equation*}
  E(\widehat H_t,J)=U_t E(\widehat H_{0t},J) U_t^{-1}.
\end{equation*}
In turn, it follows from Lemma~\ref{EVP-est} and the invariance of the subspaces~${\mathcal{W}}_n$ with respect to $\widehat H_{0t}$ that
\begin{equation*}
  E(\widehat H_{0t},J)=\sum_{k\in R}P_k,
\end{equation*}
where $R \subset R_{\overline{q}}=\{k\in\mathbb{Z}\colon \abs{\overline{n}-k}\le\overline{q}\text{ or }
\abs{\overline{n}+k}\le\overline{q}\}$ is some subset (depending on~$t$ and~$J$) and $P_k$ is the orthogonal projection onto~${\mathcal{W}}_k$.
Accordingly,
\begin{equation}\label{commu}
  [P_{\mathcal{L}},E(\widehat H_t,J)]=\sum_{k\in R}[P_{\mathcal{L}},\widetilde P_k], \;
  \text{where}\quad \widetilde P_k=U_t P_k U_t^{-1}.
\end{equation}

\begin{lemma}[\textnormal{see Sec.\!~\ref{ssA4} for the proof}]\label{SOS}
There exists an integer $d>0$ such that, for the space~$\mathcal{L}$ defined in~\eqref{cL} with this~$d$,
\begin{equation*}
  \norm{[P_{\mathcal{L}},\smash{\widetilde P_k}]}<\frac{1}{4(4\overline{q}+2)}\qquad\text{for all}\quad k\in R_{\overline{q}}
\end{equation*}
for all sufficiently small $a$. Similar estimates hold for the commutators with~$P_{{\mathcal{L}}'}$.
\end{lemma}
Since the number of terms in the sum in~\eqref{commu} does not exceed~$4\overline{q}+2$, we see that condition~(iii) holds.

\paragraph{\textbf{E.}} Consider the two-parameter family $\widehat H_{t,\tau}$ defined by the formula $\widehat H_{t,\tau}=U_{\tau t}\widehat H_{0t}U_{\tau t}^{-1}$. This is a homotopy between the ${\mathcal{L}}$-tame families $\widehat H_{0t}=\widehat H_{t,0}$ and $\widehat H_t=\widehat H_{t,1}$. We have 
\begin{equation}\label{step1}
  \operatorname{sf}_{\mathcal{L}} \widehat H_t=\operatorname{sf}_{\mathcal{L}}\widehat H_{0t} 
\end{equation}
by Theorem~\ref{th1},\,(b).

\paragraph{\textbf{F.}} It readily follows from~\eqref{Hphi}, \eqref{Hphi-ast}, and the definition of~${\mathcal{L}}$ that ${\mathcal{L}}$ is an invariant subspace of the operators $\widehat H_{0t}$. Hence the partial spectral flow of the family $\{\widehat H_{0t}\}$ along~${\mathcal{L}}$ is equal to the usual spectral flow of the restriction of this family to ${\mathcal{L}}$,
\begin{equation}\label{step2}
  \operatorname{sf}_{\mathcal{L}}\{\widehat H_{0t}\}=\operatorname{sf} \{\widehat H_{0t}\big|_{{\mathcal{L}}}\}.
\end{equation}
(Although the space~${\mathcal{L}}$ is finite-dimensional, the right-hand side need not be zero, because the restrictions of the operators $\widehat H_{00}$ and $\widehat H_{01}$ to ${\mathcal{L}}$ are not necessarily isospectral.)

\paragraph{\textbf{G.}} Now let us study the spectral flow of the Dirac operator. The same gauge transformation as in~\emph{\textbf{B}},\footnote{Strictly speaking, not exactly the same; here we deal with functions defined on $X$, while~\emph{\textbf{B}} deals with lattice functions defined on~$X_a\subset X$.} $\psi\mapsto U_t^{-1}\psi$, where $U_t$ is the operator of multiplication by $e^{iF(x,t)}$, reduces the family $\widehat D_t=D(\widehat p-{\mathrm{A}}(x,t))$ to the family $\widehat D_{0t}=D(\widehat p-{\mathrm{A}}_0(t))$,
\begin{equation*}
  \widehat D_t
  =U_t D(\widehat p-{\mathrm{A}}_0(t)) U_t^{-1}\equiv
  U_t \widehat D_{0t} U_t^{-1}.
\end{equation*}
Using the homotopy $\widehat D_{t,\tau}=U_{t\tau} \widehat D_{0t} U_{t\tau}^{-1}$, we conclude that
\begin{equation}\label{step3}
  \operatorname{sf}\{\widehat D_t\}=\operatorname{sf}\{\widehat D_{0t}\}.
\end{equation}
The vector functions
\begin{equation}\label{basis2}
  u_{mn}(x)= \begin{pmatrix}
e^{i\tfrac{\pi m}L x_1
    +i\tfrac{2\pi n}l x_2}\\
    -ie^{-i\tfrac{\pi m}L x_1
    +i\tfrac{2\pi n}l x_2}
  \end{pmatrix},\; x\in X,
  \quad m,n\in\mathbb{Z},
\end{equation}
form an orthonormal basis in~${\mathcal{H}}_0$ and satisfy the boundary conditions~\eqref{bmbc} (see Sec.\!~\ref{ssA5}). Hence they lie in the domain of the Dirac operators. The subspace
\begin{equation*}
  \widetilde{\mathcal{L}} =\operatorname{Lin}\{u_{mn}\colon m^2+n^2\le d^2\}\subset{\mathcal{H}}_0,
\end{equation*}
as well as its orthogonal complement $\widetilde{\mathcal{L}}^\perp$, is invariant with respect to~$\widehat D_{0t}$, and the restriction of $\widehat D_{0t}$ to $\widetilde{\mathcal{L}}^\perp$ is boundedly invertible (see Sec.\!~\ref{ssA6}).
Hence the spectral flow of $\{\widehat D_{0t}\}$ is equal to that of its  restriction to~$\widetilde{\mathcal{L}}$,
\begin{equation}\label{A}
  \operatorname{sf}\{\widehat D_{0t}\}=\operatorname{sf}\{\widehat D_{0t}\big|_{\widetilde{\mathcal{L}}}\}.
\end{equation}

\paragraph{\textbf{H.}}
Consider the mapping $W\colon{\mathcal{H}}\to{\mathcal{H}}_a$ given by the formula
\begin{equation*}
  W\begin{pmatrix}
     u_B \\
     u_A 
   \end{pmatrix}
   =\begin{pmatrix}
     \varphi_B \\
     \varphi_A
   \end{pmatrix},
\end{equation*}
where 
\begin{align*}
   \varphi_B(x)&=\bigl[e^{i\langle K,x\rangle}u_B(x)\bigr]\big|_{X_B},
\\
   \varphi_A(x)&=e^{-\tfrac{5\pi}6i}\bigl[e^{i\langle K,x\rangle}u_A(x)\bigr]\big|_{X_A}.
\end{align*}
This mapping can also be described by the formula  
\begin{equation*}
 W(u_{mn})=\varphi_{\overline{m}+m,\overline{n}+n}, \qquad m,n\in\mathbb{Z},
\end{equation*}
and hence its restriction to~$\widetilde{\mathcal{L}}$ (which we denote by the same letter~$W$) is an isomorphism onto the subspace~${\mathcal{L}}$.

\paragraph{\textbf{I.}} 
Since ${\mathcal{L}}$ is $\widehat H_{0t}$-invariant, it follows that the operator
\begin{equation*}
  \widehat R_t=W^{-1}\widehat H_{0t}\big|_{\mathcal{L}} W\colon\widetilde{\mathcal{L}}\longrightarrow\widetilde{\mathcal{L}}
\end{equation*}
is well defined, and 
\begin{equation}\label{B}
\operatorname{sf}\{\widehat H_{0t}\big|_{\mathcal{L}}\}=\operatorname{sf}\{\widehat R_t\}.
\end{equation}

\paragraph{\textbf{K.}} 
Now note that $\widehat R_t\to\widehat D_{0t}\big|_{\widetilde{\mathcal{L}}}$ in the operator norm uniformly with respect to $t\in[0,1]$ as $a\to0$ (see Sec.\!~\ref{ssA7}).
This also implies the resolvent convergence, because $\widetilde{\mathcal{L}}$ is finite-dimensional. It follows that
\begin{equation}\label{C}
  \operatorname{sf}\{\widehat R_t\}=\operatorname{sf}\{\widehat D_{0t}\big|_{\widetilde{\mathcal{L}}}\}
\end{equation}
for sufficiently small~$a$, because the spectral projections of $\widehat R_t$ converge to those of $\widehat D_{0t}\big|_{\widetilde{\mathcal{L}}}$ and hence the partition $0=t_0<t_1<t_2<\dotsm<t_{n+1}=1$ of the interval $[0,1]$ and the numbers $\gamma_1,\dotsc,\gamma_{n+1}\in(-\delta,\delta)$ in the definition of spectral flow can be chosen to be the same for $\{\widehat R_t\}$ and  $\{\widehat D_{0t}\big|_{\widetilde{\mathcal{L}}}\}$.

\medskip

Now we combine~\eqref{step1}, \eqref{step2}, \eqref{step3}, \eqref{A}, \eqref{B}, and~\eqref{C} and conclude that
\begin{equation*}
  \operatorname{sf}\{\widehat H_{0t}\big|_{\mathcal{L}}\}=\operatorname{sf}\{\widehat D_{0t}\}.
\end{equation*}

The proof of Theorem~\ref{th-main} is complete. \qed

\section{Partial spectral flow: Details}\label{s4}

The aim of this section is to give more insight into the notion of partial spectral flow and provide a proof of Theorem~\ref{th1}. A~key point in the concept of partial spectral flow is given by condition~(iii) in Definition~\ref{d-Ltame}, which states that the commutator of projections onto two subspaces is sufficiently small. We study some properties following from such smallness in Sec.~\ref{ss41} and then use the results in Sec.~\ref{ss42} to prove Theorem~\ref{th1}.

\subsection{Almost reducible subspaces}\label{ss41}

Let ${\mathcal{H}}$ be a Hilbert space with inner product $(\,\boldsymbol\cdot\,,\,\boldsymbol\cdot\,)$,
and let ${\mathcal{L}}\subset {\mathcal{H}}$ be a subspace. A subspace ${\mathcal{V}}\subset {\mathcal{H}}$ is said
to be \textit{reducible} (with respect to~${\mathcal{L}}$, or, more precisely,
with respect to the decomposition ${\mathcal{H}}={\mathcal{L}}\oplus {\mathcal{L}}^\perp$) if
\begin{equation*}
    {\mathcal{V}}=({\mathcal{V}}\cap {\mathcal{L}})\oplus ({\mathcal{V}}\cap {\mathcal{L}}^\perp).
\end{equation*}
This is obviously equivalent to the condition
$[P_{\mathcal{L}},P_{\mathcal{V}}]=0$, where $[A,B]=AB-BA$ is the commutator of operators~$A$ and~$B$.
\begin{definition}\label{d01}
Let $\varepsilon\ge0$. We say that a subspace ${\mathcal{V}}\subset {\mathcal{H}}$ is
$\varepsilon$-\textit{reducible} with respect to ${\mathcal{L}}$ (or simply
$\varepsilon$-\textit{reducible}, provided that ${\mathcal{L}}$ is clear from the
context) if
\begin{equation*}
    \norm{[P_{\mathcal{L}},P_{\mathcal{V}}]}\le\varepsilon.
\end{equation*}
\end{definition}
We will also say for brevity that ${\mathcal{V}}$ is \textit{almost reducible} if it is $\varepsilon$-reducible with a sufficiently small $\varepsilon$, where being ``sufficiently small'' means that $\varepsilon<\varepsilon_0$, where $\varepsilon_0>0$ depends on the context. Namely, each of the subsequent assertions is true for some $\varepsilon_0>0$, and we need all of them (or part of them) be true for almost reducible subspaces, so we just take the minimum of all the corresponding~$\varepsilon_0$.

Consider the quadratic form $A[u]=A(u,u)$ on~${\mathcal{H}}$ associated with the Hermitian form
\begin{equation}\label{e01}
    A(u,v)=(u,P_{\mathcal{L}} v)-(u,P_{{\mathcal{L}}^\perp}v)
          \equiv (u,(2P_{\mathcal{L}}-1)v).
\end{equation}
Let ${\mathcal{V}}\subset {\mathcal{H}}$ be a finite-dimensional subspace. By $A_{\mathcal{V}}[u]$ we
denote the restriction of the form $A[u]$ to ${\mathcal{V}}$.

\begin{lemma}\label{p01}
Assume that ${\mathcal{V}}$ is $\varepsilon$-reducible with $\varepsilon<\frac12$. Then the form
$A_{\mathcal{V}}$ is nonsingular, and if ${\mathcal{V}}={\mathcal{V}}_+\oplus {\mathcal{V}}_-$ is the decomposition
of ${\mathcal{V}}$ into the positive and negative subspaces of this form, then
\begin{equation}\label{e02}
    \abs{A_{\mathcal{V}}[u]}\ge (1-2\varepsilon)\norm{u}^2, \qquad u\in {\mathcal{V}}_\pm.
\end{equation}
\end{lemma}
\begin{proof}
For brevity, write $P:=P_{\mathcal{L}}$ and $Q:=P_{\mathcal{V}}$. One has
\begin{equation*}
    A_v[u]=(u,Cu),\qquad u\in {\mathcal{V}},
\end{equation*}
where the self-adjoint operator $C\colon {\mathcal{V}}\longrightarrow {\mathcal{V}}$ corresponding to
the quadratic form $A_{\mathcal{V}}$ is given by $C=Q(2P-1)$. Let $u\in {\mathcal{V}}$ be
an eigenvector of~$C$, $Cu=\lambda u$. Thus, we have
\begin{align*}
    \lambda u=Q(2P-1)u&=(2P-1)u+2[Q,P]u, \quad\text{or}
\\
    (2P-1)u&=\lambda u+2[P,Q]u.
\end{align*}
The operator $2P-1$ is unitary, and $\norm{[P,Q]}\le\varepsilon$. Hence, by
the triangle inequality,
\begin{equation*}
    \norm{u}\le\abs{\lambda}\norm{u}+2\varepsilon\norm{u}
    \quad\Longrightarrow\quad
    \abs{\lambda}\ge 1-2\varepsilon.
\end{equation*}
Since $\varepsilon<\frac12$, we see that $\lambda\ne0$ (hence the form $A_{\mathcal{V}}$ is
nonsingular) and moreover,
\begin{equation*}
    \pm A_{\mathcal{V}}[u]\ge (1-2\varepsilon)\norm{u}^2\quad\text{on ${\mathcal{V}}_\pm$.}
\end{equation*}
The proof of the lemma is complete.
\end{proof}

We see that $\varepsilon_0=\frac12$ for this lemma.

\begin{definition}\label{d02}
If a finite-dimensional subspace ${\mathcal{V}}\subset {\mathcal{H}}$ satisfies the
assumptions of Lemma~\ref{p01}, then the \textit{dimension of ${\mathcal{V}}$
along} ${\mathcal{L}}$ is defined as
\begin{equation*}
    \dim_{\mathcal{L}} {\mathcal{V}}=\sigma_+(A_{\mathcal{V}}),
\end{equation*}
where $\sigma_+(A_{\mathcal{V}})$ is the positive index of inertia of the
form~$A_{\mathcal{V}}$ (i.e., the dimension of the subspace ${\mathcal{V}}_+$ in the
decomposition of~${\mathcal{V}}$ in Lemma~\ref{p01}).
\end{definition}

\begin{lemma}\label{p02}
If a subspace ${\mathcal{V}}\subset {\mathcal{H}}$ is $\varepsilon$-reducible with respect to~${\mathcal{L}}$,
then it is $\varepsilon$-reducible with respect to~${\mathcal{L}}^\perp$. Further, if
${\mathcal{V}}$ is finite-dimensional and $\varepsilon<\frac12$, then
\begin{equation*}
    \dim_{\mathcal{L}} {\mathcal{V}}+\dim_{{\mathcal{L}}^\perp}{\mathcal{V}}=\dim {\mathcal{V}}.
\end{equation*}
\end{lemma}

\begin{proof}
It suffices to note that $P_{{\mathcal{L}}^\perp}=1-P_{\mathcal{L}}$, so that
\begin{equation*}
 [P_{\mathcal{L}},P_{\mathcal{V}}] =-[P_{{\mathcal{L}}^\perp},P_{\mathcal{V}}],
\end{equation*}
and further that ${\mathcal{V}}_+$ and ${\mathcal{V}}_-$
exchange places when we pass from $\varepsilon$-reducibility with respect
to~${\mathcal{L}}$ to that with respect to~${\mathcal{L}}_\perp$.
\end{proof}

\begin{lemma}\label{p03}
Let ${\mathcal{V}}_j\subset {\mathcal{H}}$, $j=1,2$, be orthogonal $\varepsilon$-reducible
subspaces. Then their direct sum ${\mathcal{V}}_1\oplus {\mathcal{V}}_2$ is
$2\varepsilon$-reducible. If, moreover, they are finite-dimensional and
$\varepsilon<\frac14$, then
\begin{equation}\label{e05}
    \dim_{\mathcal{L}}({\mathcal{V}}_1\oplus {\mathcal{V}}_2)=\dim_{\mathcal{L}} {\mathcal{V}}_1+\dim_{\mathcal{L}} {\mathcal{V}}_2.
\end{equation}
\end{lemma}

\begin{proof}
Let ${\mathcal{V}}={\mathcal{V}}_1\oplus {\mathcal{V}}_2$, $Q=P_{\mathcal{V}}$, and $Q_j=P_{{\mathcal{V}}_j}$, $j=1,2$. We have
$Q=Q_1+Q_2$, and so the first assertion is obvious. To prove the
second assertion, consider the subspace ${\mathcal{W}}={\mathcal{V}}_{1+}\oplus
{\mathcal{V}}_{2+}\subset {\mathcal{V}}$. We cannot claim that ${\mathcal{W}}={\mathcal{V}}_+$; however, we will
show that the restriction of the form $A_{\mathcal{V}}$ to this subspace (i.e.,
just the form $A_{\mathcal{W}}$) is positive definite. Indeed, let $u\in {\mathcal{W}}$.
Then $u=u_1+u_2$, $u_j\in {\mathcal{V}}_{j+}$, and we have
\begin{equation*}
    A[u]=A[u_1]+A[u_2]+2\operatorname{Re}(u_1,(2P-1)u_2),
\end{equation*}
where $P=P_{\mathcal{L}}$. Next,
\begin{align*}
    (u_1,(2P-1)u_2)&=(u_1,(2P-1)Q_2u_2)
\\
    &=(Q_2u_1,(2P-1)u_2)+2(u_1,[Q_2,P]u_2).
\end{align*}
The first term is zero, because $Q_2u_1=0$, and we obtain
\begin{equation*}
    \abs{(u_1,(2P-1)u_2)}\le 2\varepsilon\norm{u_1}\norm{u_2}.
\end{equation*}
Finally,
\begin{equation*}
    A[u]\ge (1-2\varepsilon)\norm{u_1}^2+(1-2\varepsilon)\norm{u_2}^2
    -4\varepsilon\norm{u_1}\norm{u_2}.
\end{equation*}
The discriminant
\begin{equation*}
    D(\varepsilon)=16\varepsilon^2-4(1-2\varepsilon)^2=16\varepsilon-4
\end{equation*}
of the quadratic form on the right-hand side is negative for
$\varepsilon<\frac14$, and hence the form $A_{\mathcal{W}}=A_{\mathcal{V}}\big|_{{\mathcal{W}}}$ itself is
positive definite. We conclude that
\begin{equation}\label{e03}
    \sigma_+(A_{\mathcal{V}})\ge\dim {\mathcal{W}}=\dim_{\mathcal{L}}{{\mathcal{V}}_1}+\dim_{\mathcal{L}}{{\mathcal{V}}_2}.
\end{equation}
The same reasoning with ${\mathcal{L}}$ and ${\mathcal{L}}_\perp$ interchanged shows that
\begin{equation}\label{e04}
    \sigma_-(A_{\mathcal{V}})\ge \dim_{{\mathcal{L}}^\perp}{{\mathcal{V}}_1}+\dim_{{\mathcal{L}}^\perp}{{\mathcal{V}}_2}.
\end{equation}
Assume that the inequality in~\eqref{e03} is strict. We
add~\eqref{e04} to~\eqref{e03} and use Lemma~\ref{p02} to obtain
\begin{align*}
    \dim {\mathcal{V}}&=\sigma_+(A_{\mathcal{V}})+\sigma_-(A_{\mathcal{V}})
\\
    &>\dim_{\mathcal{L}}{{\mathcal{V}}_1}+\dim_{\mathcal{L}}{{\mathcal{V}}_2}
    +\dim_{{\mathcal{L}}^\perp}{{\mathcal{V}}_1}+\dim_{{\mathcal{L}}^\perp}{{\mathcal{V}}_2}
\\
    &=\dim {\mathcal{V}}_1+\dim {\mathcal{V}}_2=\dim {\mathcal{V}},
\end{align*}
which is a contradiction. Thus, we have the equality
in~\eqref{e03}, relation~\eqref{e05} holds, and the proof of the
lemma is complete.
\end{proof}

\begin{lemma}\label{p04}
Let ${\mathcal{V}}_t$, $t\in[a,b]$, be a continuous family of
finite-dimensional $\varepsilon$-reducible subspaces, where $\varepsilon\le\frac12$
and the continuity is understood as the norm continuity of the
corresponding family of projections $Q(t)=P_{{\mathcal{V}}_t}$. Then
$\dim_{\mathcal{L}} {\mathcal{V}}_t$ is independent of~$t\in[a,b]$.
\end{lemma}

\begin{proof}
It is well known that there exists a unitary $U(t)$ continuously
depending on~$t$ such that the space ${\mathcal{V}}=U(t){\mathcal{V}}_t$ is independent
of~$t$. The operator
\begin{equation*}
  C(t)=U(t)Q(t)(2P-1)U^{-1}(t)\colon {\mathcal{V}}\longrightarrow {\mathcal{V}},
\end{equation*}
which determines the form $A_{{\mathcal{V}}_t}$ transferred by $U(t)$ to the
fixed subspace~${\mathcal{V}}$, continuously depends on~$t$ and is nonsingular
for all~$t$. Hence $\sigma_+(A_{{\mathcal{V}}_t})=\operatorname{const}$, as desired. The proof
of the lemma is complete.
\end{proof}

\subsection{Proof of Theorem~\ref{th1}}\label{ss42}

(a) We need to prove that the right-hand side of~\eqref{e06} is independent of the choice of the partition of the interval $[0,1]$ and the numbers~$\gamma_j$. To compare two such choices, it suffices to consider the case in which both partitions are the same (just take a new partition containing the points of both). Further, we can change the numbers $\gamma_j$ one by one, so it suffices to see what happens if we change just one of them, i.e., replace $\gamma_j$ by some $\widetilde\gamma_j$ on the interval $[t_{j-1},t_j]$. The points $\gamma_j$ and $\widetilde\gamma_j$ do not lie in the spectrum of $B_t$ for any $t\in[t_{j-1},t_j]$.  The projection onto the linear span~${\mathcal{V}}(B_t,\gamma_j,\widetilde\gamma_j)$ of eigenvectors of~$B_t$ corresponding to eigenvalues lying between~$\gamma_j$ and~$\widetilde\gamma_j$ can be expressed as the contour integral of the resolvent of~$B_t$ over a loop crossing the real line at the points~$\gamma_j$ and~$\widetilde\gamma_j$ (see Fig.~\ref{fig08}) and hence continuously depends on $t\in[t_{j-1},t_j]$.
\begin{figure}[ht]
\centering
\includegraphics{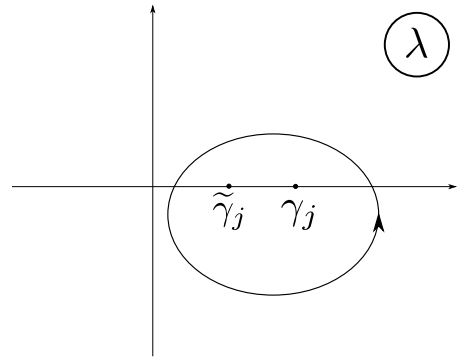}
\caption{Integration contour for the projection onto~${\mathcal{V}}(B_t,\gamma_j,\widetilde\gamma_j)$}
\label{fig08}
\end{figure}
In other words, ${\mathcal{V}}(t)={\mathcal{V}}(B_t,\gamma_j,\widetilde\gamma_j)$ depends on $t$ continuously on that interval, and $\dim_{\mathcal{L}}{\mathcal{V}}(t_{j-1})=\dim_{\mathcal{L}}{\mathcal{V}}(t_j)$ by Lemma~\ref{p04}. Now let us see what changes occur in the sum~\eqref{e06} when replacing $\gamma_j$ by $\widetilde\gamma_j$. Only the $(j-1)$st and $j$th terms are affected; the number $\dim_{\mathcal{L}} {\mathcal{V}}(t_{j-1})=\dim_{\mathcal{L}}{\mathcal{V}}(t_j)$ is added to one of these terms and subtracted from the other by Lemma~\ref{p03}, and so the sum remains unchanged.

The ${\mathcal{L}}^\perp$-tameness is a straightforward consequence of
Lemma~\ref{p02}, and formula~\eqref{compl} follows from the fact that the sum of positive and negative indices of inertia of a nondegenerate quadratic form is the total dimension of the space where the form is considered. The proof of~(a) is complete.

\medskip

(b) It suffices to prove that
\begin{multline*}
    \operatorname{sf}_{\mathcal{L}}\{B(t,0)\}
    +\operatorname{sf}_{\mathcal{L}}\{B(1,t)\}
\\
-\operatorname{sf}_{\mathcal{L}}\{B(t,1)\}
    -\operatorname{sf}_{\mathcal{L}}\{B(0,t)\}=0.
\end{multline*}
The left-hand side of this equation is just the partial spectral
flow along ${\mathcal{L}}$ of the family obtained by the restriction of
$B(t,\tau)$ to the boundary of the unit square (with the
counterclockwise sense). The closed contour (the boundary) can be
contracted into a point within the unit square, without changing
the partial spectral flow. (Indeed, for sufficiently small changes of the
contour the partition of the interval and the $\gamma_j$ can remain
unchanged, so the partial spectral flow remains constant.) The
partial spectral flow of the constant family is obviously zero, and
the theorem follows.

The proof of Theorem~\ref{th1} is complete. \qed

\begin{remark}
Condition~\eqref{weak-iso} is a generalization of the isospectrality
condition, which looks as follows for the case of partial spectral
flow:

\textit{For any
$\gamma,\widetilde\gamma\in(-\delta,\delta)$, where $\delta>0$ is the same as in Definition~\textup{\ref{d-Ltame}},
one has
\begin{equation}\label{isol}
    \dim_{\mathcal{L}} {\mathcal{V}}(B(0,\tau),\gamma,\widetilde\gamma)=\dim_{\mathcal{L}} {\mathcal{V}}(B(1,\tau),\gamma,\widetilde\gamma)
\end{equation}
for all $\tau\in[0,1]$.}

In the case of the usual spectral flow (${\mathcal{L}}={\mathcal{H}}$), this condition becomes the common isospectrality condition: the spectra of $B(0,\tau)$ and $B(1,\tau)$ in a neighborhood of $\lambda=0$ are the same for each~$\tau$. Although condition~\eqref{weak-iso} is much weaker than the
isospectrality condition~\eqref{isol}, it is sufficient for the
homotopy invariance to hold.
\end{remark}

\section{Conclusions}

To conclude, let us look at our results from a more general point of view. Transfer of concepts between condensed matter physics and the ``fundamental physics'' such as high energy physics, cosmology, etc. is an important source of innovations in both of the fields. It is probably enough to mention such concepts as spontaneous broken symmetry and renormalization group which revolutionized the fields. To be closer to our specific subject one can just refer to the role of graphene as ``CERN on the desk'', with long-waiting physical realizations of Klein paradox and relativistic atomic collapse \cite{Kats12}. There is however a fundamental difference: while in high-energy physics and quantum field theory the space-time is assumed to be continuous (despite the use of lattices is an extremely useful technical tool \cite{Creu}), in condensed matter physics the discreteness of crystal lattices is the crucial fact. The difference is especially important when transferring {\it topological} concepts to condensed matter physics: from the point of view of topology, continuum and a discrete lattice are {\it dramatically} different. In this paper we have demonstrated, using a specific simple example, that in some cases this transfer can be rigorously justified. Namely, one can make a conclusion that under certain circumstances adiabatically growing magnetic fluxes will induce electron-hole pair creation in graphene, because of nonvanishing spectral flow of Dirac operator \cite{KatNa1}. The spectral flow of the tight-binding Hamiltonian at honeycomb lattice {\it is} obviously zero but nevertheless the physical conclusion formulated above is still valid and can be justified via the new concept of {\it partial} spectral flow. Despite globally the (unbounded and differential) Dirac operator and (bounded and finite-matrix) Hamiltonian on honeycomb lattices are completely different their topological properties are connected in some nontrivial way. We believe that this example can be interesting for a much more general issue on the connections between lattice and continuous models in physics.

\section*{Acknowledgements}

The work of MIK was supported by the JTC-FLAGERA Project GRANSPORT. The work of VN was supported by the Ministry of Science and Higher Education of the Russian Federation within the framework of the Russian State Assignment under contract No.~AAAA-A20-120011690131-7.

\appendix
\section{Some technical computations}\label{AppA}

\def\appendixname{\!}
\def\sectcounterend{}

\subsection{Proof of Lemma~\ref{l202}}\label{ssA1}

Consider the mapping $\omega\colon\mathbb{R}^2\to\mathbb{R}^2$, $(x_1,x_2)\mapsto(-x_1,x_2)$, and the rectangle
\begin{equation*}
  \widetilde X=X\cup\omega(X)=[L,L]\times[0,l],
\end{equation*}
which we identify with the torus obtained by pasting together the endpoints of each of the two intervals. The lattice $\widetilde X_B=X_B\cup\omega(X_A)$ is the natural extension of the lattice~$X_B$ from~$X$ to~$\widetilde X$, and the mapping $V\colon{\mathcal{H}}_a\to\ell^2(\widetilde X_B)$ given by
\begin{equation*}
  [Vf](x)=\begin{cases}
            f(x),&x_1>0,\\
            f(\omega(x))\equiv f(-x_1,x_2),&x_1<0,
          \end{cases}
          \quad x\in\widetilde X_B,
\end{equation*}
is a unitary isomorphism. Note that 
\begin{equation}\label{Vphimn}
  [V\varphi_{mn}](x)=e^{i\tfrac{\pi m}L x_1
    +i\tfrac{2\pi n}l x_2},\qquad x\in\widetilde X_B,
\end{equation}
and so it suffices to prove that the functions~\eqref{Vphimn}, where $(m,n)\in G_0$, form an orthonormal basis in $\ell^2(\widetilde X_B)$. To show this, we reduce~$G_0$ to a more convenient indexing set. Consider the vectors
\begin{equation}\label{e12}
  e_1=(2M,N),\qquad e_2=(2M,-N). 
\end{equation}
One can show by a straightforward computation that the functions $\varphi_{mn}(x)$, $x\in X_a$, obey the transformation rule
\begin{equation}\label{trans-rule}
\begin{aligned}
  \varphi_{\widetilde m\widetilde n}(x)&=e^{-\tfrac{2\pi(j+k)i}3}\varphi_{mn}(x)
\\ 
\text{if}\quad
  (\widetilde m,\widetilde n)&=(m,n)+je_1+ke_2,\quad j,k\in\mathbb{Z},
\end{aligned}
\end{equation}
and so do the functions~\eqref{Vphimn}; hence we can transform~$G_0$ by shifting each element $(m,n)\in G_0$ by some vector of the integer lattice generated by~$e_1$ and~$e_2$. It is an elementary but tiresome exercise to show that such shifts can be used to reduce $G_0$ to the set
$G_1=\{(m,n)\colon -M\le m<M ,\;-N\le n<N\}$. Since the lattice $\widetilde X_B$ on the torus~$\widetilde X$ is the (skew) product of two one-dimensional lattices on circles with $2M$ and $2N$ points, respectively, it readily follows that the functions~\eqref{Vphimn} with $(m,n)\in G_1$ (and hence with $(m,n)\in G_0$) indeed form an orthonormal basis in~$\ell(\widetilde X_B)$. \qed

\subsection{Action of $\widehat H_{0t}$ on basis vectors}\label{ssA2}

We have 
\begin{equation*}
  \widehat H_{0t}=H(\widehat p-{\mathrm{A}}_0(t))
  =\frac{2}{3a}
  \begin{pmatrix}
    0 & T(\widehat p-{\mathrm{A}}_0(t)) \\
    T^*(\widehat p-{\mathrm{A}}_0(t)) & 0
  \end{pmatrix}.
\end{equation*}
The basis functions $\varphi_{mn}(x)$ given by~\eqref{basis1} agree with the boundary conditions~\eqref{bc} in the sense that the values of components of these functions prescribed by the boundary conditions at the fictitious nodes outside~$X$ are given by the same exponential expressions as the components themselves. As a consequence, the application of a function of $\widehat p$ to these components amounts to the replacement of $\widehat p$ by the corresponding wave number. In particular, we have
\begin{align*}
 \widehat H_{0t}\varphi_{mn}&=\frac{2}{3a}
 \begin{pmatrix}
   T\biggl(-\dfrac{\pi m}{L},\dfrac{2\pi(n-q(t)}{l}\biggr)
   e^{-i\tfrac{\pi m}L x_1
    +i\tfrac{2\pi n}l x_2}\\[12pt]
   T^*\biggl(\dfrac{\pi m}{L},\dfrac{2\pi(n-q(t)}{l}\biggr)
   e^{i\tfrac{\pi m}L x_1
    +i\tfrac{2\pi n}l x_2}  
 \end{pmatrix}
\\
 &=\frac{2}{3a}
 T\biggl(-\dfrac{\pi m}{L},\dfrac{2\pi(n-q(t)}{l}\biggr)
 \varphi_{-m,n}
\\
 &=\frac{2}{3a}e^{\tfrac{4\pi i}3}
 T\biggl(-\dfrac{\pi m}{L},\dfrac{2\pi(n-q(t)}{l}\biggr)
 \varphi_{2\overline{m}-m,n}
 \\
 &\equiv \mu(m,n,t)\varphi_{2\overline{m}-m,n},
\end{align*}
because $T^*(p_1,p_2)=T(-p_1,p_2)$ and in view of the transformation formula~\eqref{trans-rule}. (Note that $\overline{m}=2M$ and so $(2\overline{m}-m,n)=(-m,n)+e_1+e_2$.) A straightforward computation using definition~\eqref{tbh} of~$T(p)$ shows that
\begin{equation}\label{mu2}
  \mu(m,n,t)=\frac{2}{3a}e^{-\tfrac13i\alpha}
  (e^{i\alpha}-2\cos\beta),
\end{equation}
where
\begin{equation*}
  \alpha=\frac{\pi(m-\overline{m})}{2M},\qquad
  \beta=\frac{\pi(n-q(t))}{N}.
\end{equation*} 
Further,
\begin{equation*}
  \widehat H_{0t}\varphi_{2\overline{m}-m,n}=\mu(2\overline{m}-m,n,t)\varphi_{mn},
\end{equation*}
and we note that
\begin{equation*}
  \frac{\pi((2\overline{m}-m)-\overline{m})}{2M}=\frac{\pi(\overline{m}-m)}{2M}=-\alpha
\end{equation*}
and hence $\mu(2\overline{m}-m,n,t)=\mu^*(m,n,t)$.

\subsection{Proof of Lemma~\ref{EVP-est}}\label{ssA3}

First, consider a subspace ${\mathcal{V}}_{mn}$ with $m>\overline{m}$, $(m,n)\in G$. In this case,
\begin{equation*}
  \frac\pi2\ge\alpha=\frac{\pi(m-\overline{m})}{2M}\ge
  \frac{\pi}{2M}=\frac{3\pi a}{2L}.
\end{equation*}
Consequently, 
\begin{equation*}
  \abs{\mu(m,n,t)}\ge\frac{2}{3a}\operatorname{Im} e^{i\alpha}
  \ge\frac{2}{3a}\sin\frac{3\pi a}{2L}
  \xrightarrow{a\to0}\frac\pi L,
\end{equation*}
and the right-hand side is greater than $\pi/(2L)$ for small~$a$.

Next, consider the subspace~${\mathcal{W}}_n$. The eigenvalue in question has the form
\begin{equation}\label{vaz}
\begin{aligned}
  \mu(\overline{m},n,t)&=\frac{2}{3a}(1-2\cos\beta)
  \\
  &=\frac{2}{3a}
  \biggl[1-2\cos\biggl[\pm\frac{\pi}{3}+\frac\pi N(n\pm\overline{n}-q(t))\biggr]\biggr].
\end{aligned}
\end{equation}
(Recall that $N=3\overline{n}$.) If $\abs{n-q(t)\pm\overline{n}}\ge1$, then
\begin{equation*}
  \frac{2\pi}{3}+\frac{\pi\abs{q(t)}a\sqrt3}l\ge
  \left\vert\beta\pm\frac\pi3\right\vert\ge\frac{\pi a\sqrt3}l,
\end{equation*}
and hence
\begin{equation*}
  \abs{\mu(\overline{m},n,t)}\ge \frac{2}{3a}
  \frac{\pi a\sqrt3}l\sin\frac\pi3
  =\frac\pi l 
\end{equation*}
for sufficiently small~$a$. Now we see that it suffices to set
\begin{equation*}
  \delta=\min\biggl\{\frac\pi{2L},\frac\pi l\biggr\}, \qquad
  \overline{q}=\max_{t\in[0,1]}\abs{q(t)}
\end{equation*}
and take a sufficiently small~$a_0$. The proof of the lemma is complete. \qed

\subsection{Orthogonal basis in~${\mathcal{H}}_0$}\label{ssA5}

By analogy with Sec.~\ref{ssA1},
the mapping $V_0\colon{\mathcal{H}}_0\to L^2(\widetilde X)$ given by
\begin{equation*}
  [V_0f](x)=\begin{cases}
            f(x),&x_1\ge 0,\\
            f(\omega(x))\equiv f(-x_1,x_2),&x_1<0,
          \end{cases}
          \quad x\in\widetilde X,
\end{equation*}
is a unitary isomorphism. Further, the exponentials
\begin{equation}\label{emn}
 [V_0u_{mn}](x)=e^{i\tfrac{\pi m}L x_1
    +i\tfrac{2\pi n}l x_2}=:e_{mn}(x),\qquad x\in\widetilde X,
\end{equation}
with $(m,n)\in\mathbb{Z}^2$ form an orthonormal basis in $L^2(\widetilde X)$. Hence the original functions $u_{mn}(x)$ form an orthonormal basis in ${\mathcal{H}}_0$, as desired. Further, the boundary conditions~\eqref{bmbc} require that $iu_A(x_1,x_2)=u_B(x_1,x_2)$ for $x_1=0$ and $x_1=L$. For the functions $u_{mn}(x)$, this amounts to the requirement that
\begin{equation*}
  e^{i\tfrac{\pi m}L x_1}=e^{-i\tfrac{\pi m}L x_1} 
\end{equation*}
for $x_1=0$ and $x_1=L$, which is obviously true.

\subsection{Proof of Lemma~\ref{SOS}}\label{ssA4}

The proof is based on some properties of the function $e^{iF(x,t)}$ occurring in the definition of the operator~$U_t$ (Sec.~\ref{ss34},~\textbf{\emph{B}}). Consider the expansion of the function $e^{iF(x,t)}$ restricted to the lattice~$X_a$ in the basis functions $\varphi_{mn}$:
\begin{equation*}
  e^{iF(x,t)}=\sum_{(m,n)\in G_0}
  b(m,n,t)\varphi_{mn}(x), x\in X_a.
\end{equation*}
We need estimates for the coefficients $b(m,n,t)$. To state these estimates, we introduce the function
\begin{equation*}
  \rho(m,n)=\min_{j,k\in\mathbb{Z}}
  [(m+2M(j+k))^2+(n+N(j-k))^2]^{1/2}.
\end{equation*}
This function is none other than the distance from $(m,n)$ to the nearest point of the integer lattice generated by the vectors $e_1=(2M,N)$ and $e_2=(2M,-N)$ (see~\eqref{e12}). 
\begin{proposition}\label{p05}
There exists a constant $C>0$ independent of~$a$ such that
\begin{equation}\label{crucial}
  \abs{b(m,n,t)}\le\frac{C}{1+\rho^2(m,n)}.
\end{equation}
\end{proposition}
\begin{proof}
Let us continue the function $e^{iF(x,t)}$ from $X$ to~$\widetilde X$ as an even function of~$x_1$. We use the same notation for the continuation as for the function itself. Thus,
\begin{equation*}
  e^{iF(x_1,x_2,t)}=e^{iF(-x_1,x_2,t)},\qquad x\in\widetilde X.
\end{equation*}
We view $\widetilde X$ as a torus. Let
\begin{equation}\label{cmnt}
  c(m,n,t)=\frac{1}{2lL}\iint_{\widetilde X} e^{-i\tfrac{\pi m}L x_1
    -i\tfrac{2\pi n}l x_2}e^{iF(x_1,x_2,t)}\,dx_1\,dx_2
\end{equation}
be the Fourier coefficients of the function $e^{iF(x,t)}$ in the system of exponentials $\{e_{mn}(x)\}$, $(m,n)\in\mathbb{Z}^2$.
These coefficients satisfy the estimates 
\begin{equation*}
  \abs{c(m,n,t)}\le\frac{C_1}{(1+m^2)(1+n^2)}, \qquad (m,n)\in\mathbb{Z}^2,
\end{equation*}
which can be derived in a standard way by integration by parts in~\eqref{cmnt} with respect to $x_1$ for $m\ne0$ and with respect to $x_2$ for $n\ne0$. The function $e^{iF(x_1,x_2,t)}$ is continuous, but its first derivative with respect to~$x_1$ may have jump discontinuities at $x_1=0$ and $x_1=L$. Hence we can integrate at most twice by parts with respect to~$x_1$: the second time we get the integrated term, and the factor $(1+m^2)^{-1}$ cannot be improved further. We can integrate as many times as desired with respect to~$x_2$, but we just do not need a better estimate than $(1+n^2)^{-1}$. In view of the construction in Sec.~\ref{ssA1}, the coefficients $b(m,n,t)$ coincide with the coefficients in the expansion of the restriction of $e^{iF(x_1,x_2,t)}$ to~$\widetilde X_B$ in the functions~\eqref{Vphimn}, $(m,n)\in G_0$. In view of the transformation rule~\eqref{trans-rule} for the functions~\eqref{Vphimn}, we have
\begin{equation*}
  b(m,n,t)=\sum_{j,k=-\infty}^{\infty} e^{i\tfrac{2(j+k)\pi}3} c\bigl(m+2M(j+k),n+N(j-k)\bigr),
\end{equation*}
and so
\begin{multline}\label{sum}
  \abs{b(m,n,t)}\\
  \le\sum_{j=-\infty}^{\infty}
  \sum_{\substack{k=-\infty\\k=j\bmod2}}^{^\infty}
  \frac{C_1}{(1+(m+2Mj)^2)(1+(n+Nk)^2)}.
\end{multline}
Note that $M\le m\le 3M$ and $-N\le n<N$; hence we have
\begin{align*}
  \abs{m+2Mj}&\ge \frac12M\abs{j}\quad\text{for $j\notin\{-1,0\}$,}\\
  \abs{n+Nk}&\ge \frac12N\abs{k}\quad\text{for $k\notin\{-1,0,1\}$.}
\end{align*}
Now we split the sum~\eqref{sum} into four sums $\Sigma_1+\Sigma_2+\Sigma_3+\Sigma_4$, where the summation is

over the set $\Delta_1$: $j\in\{-1,0\}$ and $k\in\{-1,0,1\}$ for $\Sigma_1$;

over the set $\Delta_2$: $j\in\{-1,0\}$ and $k\notin\{-1,0,1\}$ for $\Sigma_2$;

over the set $\Delta_3$: $j\notin\{-1,0\}$ and $k\in\{-1,0,1\}$ for $\Sigma_3$;

over  the set $\Delta_4$: $j\notin\{-1,0\}$ and $k\notin\{-1,0,1\}$ for $\Sigma_4$.

\noindent Of course, we also have in mind the condition $k=j\bmod2$.

The sum~$\Sigma_1$ contains three terms, 
\begin{equation*}
  \Delta_1=\{(-1,-1),(0,0),(-1,1)\}.
\end{equation*}
Since
\begin{align*}
  (1+(m&+2Mj)^2)(1+(n+Nk)^2)
\\  
  &= 1+(m+2Mj)^2 +(n+Nk)^2
\\
  &\qquad\qquad{}+(m+2Mj)^2(n+Nk)^2
\\
  &\ge 1+(m+2Mj)^2 +(n+Nk)^2 \ge 1+\rho^2(m,n) 
\end{align*}
for any $j=k\bmod2$, we have
\begin{equation*}
  \Sigma_1\le\frac{3C_1}{1+\rho^2(m,n)}
\end{equation*}
Further,
\begin{align*}
  \Sigma_2 &\le 4C_1\sum_{k\in\mathbb{Z}\setminus\{0\}} \frac{1}{N^2k^2} = \frac{C_2}{N^2}  \\
  \Sigma_3 &\le 6C_1\sum_{j\in\mathbb{Z}\setminus\{0\}} \frac{1}{(2M)^2j^2} = \frac{C_3}{M^2}\\
  \Sigma_4 &\le C_1\sum_{j,k\in\mathbb{Z}\setminus\{0\}}
  \frac{1}{N^2k^2}\frac{1}{M^2j^2}=\frac{C_4}{M^2N^2}.
\end{align*}
Since the ratio $M/N$ is equal to $L/(l\sqrt3)$ and does not vary as $a\to0$ and $M,N\to\infty$, we readily see that there exists a constant $C_5$ such that 
\begin{equation*}
  1+\rho^2(m,n)\le C_5M=C_5 L/(l\sqrt3)N
\end{equation*}
for any $(m,n)$. Hence we arrive at~\eqref{crucial}. 
The proof of the proposition is complete.\qed  
\end{proof}

Now we can prove the lemma. One has
\begin{align}\label{1}
  [P_{\mathcal{L}},\widetilde P_k]&=P_{\mathcal{L}} U_t P_k U_t^{-1}- U_t P_k U_t^{-1}P_{\mathcal{L}}
\\\label{2}
  &=U_t P_k U_t^{-1}(1-P_{\mathcal{L}})-(1-P_{\mathcal{L}}) U_t P_k U_t^{-1},
\end{align}
and hence
\begin{align}\label{1a}
  \norm{[P_{\mathcal{L}},\smash{\widetilde P_k}]}&\le
  2\norm{P_{\mathcal{L}} U_t P_k },
\\\label{2a}
  \norm{[P_{\mathcal{L}},\smash{\widetilde P_k}]}&\le
  2\norm{(1-P_{\mathcal{L}})U_t P_k }.
\end{align}
If $Q$ is a projection, then 
\begin{equation*}
  Q U_t P_k u=(\varphi_{\overline{m} k}, u)Q (e^{iF(x,t)}\varphi_{\overline{m} k}(x)),
\end{equation*}
and hence
\begin{equation*}
  \norm{Q U_t P_k}=\norm{Q (e^{iF(x,t)}\varphi_{\overline{m} k}(x))}.
\end{equation*}

We will use the estimate~\eqref{1a} if $\abs{k+\overline{n}}\le\overline{q}$ and the estimate~\eqref{2a} if $\abs{k-\overline{n}}\le\overline{q}$. Consider the latter case. We have
\begin{equation*}
  (1-P_{\mathcal{L}})(e^{iF(x,t)}\varphi_{\overline{m} k}(x))
  =\sideset{}{'}\sum_{(m,n)\in G_0}
  b(m,n,t)\varphi_{m+\overline{m},n+k}(x),
\end{equation*}
where the prime indicates that the sum is over $(m,n)\in G_0$ satisfying $\rho(m,n+k-\overline{n})>d$. (Indeed, $1-P_{\mathcal{L}}$ annihilates any basis function~$\varphi_{js}$ with $\rho(j-\overline{m},s-\overline{n})\le d$.) If $\rho(m,n+k-\overline{n})>d$, then, by the triangle inequality for the metric generated by~$\rho$,
\begin{equation*}
  \rho(m,n)\ge \rho(m,n+k-\overline{n})-\rho(0,k-\overline{n})>d-\overline{q},
\end{equation*}
and we have
\begin{align*}
  &\norm{(1-P_{\mathcal{L}})(e^{iF(x,t)}\varphi_{\overline{m} k}(x))}^2
\\ &\le
  \sum_{\substack{(m,n)\in G_0\\ \rho(m,n)>d-\overline{q}}}
  \abs{b(m,n,t)}^2
\le C\sum_{\substack{(m,n)\in G_0\\\rho(m,n)> d-\overline{q}}}\frac{1}{(1+\rho^2(m,n))^2}
\\
  &\le C\sum_{\substack{(m,n)\in \mathbb{Z}^2\\\overline{m}^2+n^2> (d-\overline{q})^2}}\frac{1}{(1+m^2+n^2)^2}.
\end{align*}
(The last transition can be explained as follows: we shift each points of $G_0$ by some integer linear combination of $e_1$ and $e_2$ so as to ensure that $\rho^2(m,n)=m^2+n^2$ and then extend the summation to all $(m,n)\in\mathbb{Z}^2$ with $m^2+n^2>(d-\overline{q})^2$ by adding infinitely many positive terms to the sum.) Since the series $\sum(1+m^2+n^2)^{-2}$ converges, we can find $d$ such that the right-hand side of the last inequality is less than $4^{-1}(4\overline{q}+2)^{-1}$.

The proof for the case of $\abs{k+\overline{n}}\le\overline{q}$ goes along the same lines. Here we use formula~\eqref{1a} instead of~\eqref{2a}, and the role of $d-\overline{q}$ is now played by $2\overline{n}-d-\overline{q}$ (where $d$ has already be computed in the preceding case). Since $\overline{n}\to\infty$ as $a\to0$, it remains to take $a$ small enough that $2\overline{n}-d-\overline{q}>d-\overline{q}$, i.e., $\overline{n}>d$.

The proof of Lemma~\ref{SOS} is complete. \qed

\subsection{Decomposition of $\widehat D_{0t}$}\label{ssA6}

The symbol of the operator $\widehat D_{0t}$ has the form
\begin{equation*}
  D_{0t}(p)=\begin{pmatrix}
      0 & p_1+ip_2-\frac{2\pi i q(t)}{l} \\
      p_1-ip_2+\frac{2\pi i q(t)}{l} & 0
    \end{pmatrix}.
\end{equation*}
Using this expression, one can readily compute
\begin{equation}\label{D0t}
\begin{aligned}
  \widehat D_{0t}u_{mn}&=\mu_0(m,n,t)u_{-m,n},
\\
  \widehat D_{0t}u_{-m,n}&=\mu_0^*(m,n,t)u_{m,n},
\end{aligned}
\end{equation}
where
\begin{equation*}
  \mu_0(m,n,t)=\frac{2\pi(n-q(t))}{l}+i\frac{\pi m}{L}.
\end{equation*}
We see that the space~${\mathcal{H}}_0$ splits into the direct sum of two-dimensional invariant subspaces spanned by $u_{mn}$ and $u_{-m,n}$ for $m>0$ and one-dimensional invariant subspaces spanned by $u_{0n}$. The eigenvalues are $\pm\abs{\mu_0(m,n,t)}^2\ne0$ on the two-dimensional subspaces and
\begin{equation*}
  \mu_0(0,n,t)=\frac{2\pi(n-q(t))}{l}
\end{equation*}
on the one-dimensional subspaces. The latter are obviously nonzero if $\abs{n}>\overline{q}$. Since $d>\overline{q}$, it follows that all eigenvectors corresponding to zero eigenvalues lie in the space
\begin{equation*}
  \widetilde{\mathcal{L}} =\operatorname{Lin}\{u_{mn}\colon m^2+n^2\le d^2\}\subset{\mathcal{H}}_0,
\end{equation*}
which is obviously invariant, because it contains $u_{mn}$ and $u_{-m,n}$ simultaneously. One can readily show that the operator $\widehat D_{0t}$ is invertible on~$\widetilde{\mathcal{L}}^\perp$.

\subsection{Convergence of the tight-binding Hamiltonian to the Dirac Hamiltonian on~$\widetilde{\mathcal{L}}$}\label{ssA7}

It follows from the formula 
\begin{equation*}
  \widehat H_{0t}\varphi_{mn}=\mu(m,n,t)\varphi_{2\overline{m}-m,n}
\end{equation*}
(see~\eqref{Hphi}) for the tight-binding Hamiltonian and the formula 
\begin{equation*}
 W(u_{mn})=\varphi_{\overline{m}+m,\overline{n}+n}
\end{equation*}
for the isomorphism $W\colon\widetilde{\mathcal{L}}\to{\mathcal{L}}$ that the operator 
\begin{equation*}
 \widehat R_t=W^{-1}\widehat H_{0t}\big|_{\mathcal{L}} W
\end{equation*}
acts by the formula
\begin{equation*}
  \widehat R_t u_{mn}=\mu(m+\overline{m},n+\overline{n},t)u_{-m,n}.
\end{equation*}
Thus, by~\eqref{D0t}, to prove the uniform convergence $\widehat R_t\to\widehat D_{0t}$ as $a\to0$, it suffices to prove that
\begin{equation*}
  \mu(m+\overline{m},n+\overline{n},t)\xrightarrow{a\to0}\mu_0(m,n,t)
\end{equation*}
uniformly with respect to~$t\in[0,1]$. 

We have (see~\eqref{mu2})
\begin{align*} 
  \mu(m+\overline{m},n+\overline{n},t)&=\frac{2}{3a}e^{-\tfrac13i\alpha}
  (e^{i\alpha}-2\cos\beta),
\\
  \alpha&=\frac{\pi m}{2M}=3a\frac{\pi m}{2L},
\\
  \beta&=\frac\pi3+\frac{\pi(n-q(t))}{N}
\\
  &=\frac\pi3+\sqrt3a\frac{\pi(n-q(t))}{l}.
\end{align*}
We take the first term of the Taylor series as $a\to0$ and obtain
\begin{align*}
  \mu(m&+\overline{m},n+\overline{n},t)
\\
  &=\frac{2}{3a}
  \biggl[1-ia\frac{\pi m}{2L}\biggr]
  \biggl[1+3ia\frac{\pi m}{2L}-1+3a\frac{\pi(n-q(t))}{l}\biggr]
\\
  &\qquad\qquad{}+O(a)
\\
  &=\frac{2\pi(n-q(t))}{l}+i\frac{\pi m}{L}+O(a)=
  \mu_0(m,n,t)
  +O(a).
\end{align*}
Thus, we have arrived at the desired result.

\clearpage

\end{document}